\documentclass[11pt]{article}
\pdfoutput=1 
\usepackage[T1]{fontenc}
\usepackage{lmodern}
\usepackage{booktabs}
\usepackage[margin=1in]{geometry}
\usepackage{amsfonts,amssymb,amsmath,amsthm,mathtools,thmtools}
\usepackage{relsize}
\usepackage{microtype,xspace,bm,floatrow}
\usepackage{braket}
\usepackage{mdframed}
\usepackage{dsfont}

\usepackage{doi}
\usepackage{hyperref}
\usepackage[dvipsnames]{xcolor}
\hypersetup{
	colorlinks=true,
	urlcolor=MidnightBlue,
	linkcolor=MidnightBlue,
	citecolor=MidnightBlue,
	unicode
}

\usepackage{authblk}

\newcommand{\negl}{\mathsf{negl}}
\newcommand{\Tr}{\mathrm{Tr}}
\newcommand{\decode}{\mathsf{Decode}}
\newcommand{\RS}{\mathrm{RS}}

\newcommand{\bits}{\{0,1\}}
\newcommand{\bitsn}{\{0,1\}^n}

\usepackage{enumitem}
\setlist{  
  listparindent=\parindent,
  parsep=0pt,
}

\usepackage[capitalise,nameinlink,noabbrev]{cleveref}
\crefname{equation}{}{}

\usepackage[english]{babel}
\addto\extrasenglish{}
\addto\extrasenglish{}

\usepackage{lpic}

\makeatletter
\DeclareRobustCommand\bfseries{%
  \not@math@alphabet\bfseries\mathbf
  \fontseries\bfdefault\selectfont\boldmath}
\makeatother

\declaretheorem[name=Theorem]{theorem}
\declaretheorem[name=Proposition, sibling=theorem]{proposition}
\declaretheorem[name=Lemma,sibling=theorem]{lemma}
\declaretheorem[name=Corollary, sibling=theorem]{corollary}

\declaretheorem[name=Definition, sibling=theorem, style=definition]{definition}

\declaretheorem[name=Remark,sibling=theorem]{remark}

\declaretheorem[name=Problem]{problem}

\DeclareMathOperator{\poly}{poly}

\newcommand{\newclass}[2]{\newcommand{#1}{{\text{\upshape\sffamily #2}}\xspace}}
\newclass{\NP}{NP}
\newclass{\cA}{A}
\newclass{\cB}{B}
\newclass{\BPP}{BPP}
\newclass{\BQP}{BQP}
\newclass{\TFNP}{TFNP}

\newcommand{\newprob}[2]{\newcommand{#1}{{\text{\upshape\scshape #2}}\xspace}}
\newprob{\CI}{Code Intersection}
\newprob{\CIshort}{CI}
\newprob{\OPI}{Optimal Polynomial Intersection}
\newprob{\OPIshort}{OPI}

\renewcommand{\H}{\mathbf{H}}
\renewcommand{\Pr}{\mathbf{Pr}}

\newcommand{\val}{\mathrm{val}}
\newcommand{\OPIval}{\mathrm{val}_{\OPIshort}}

\let\OLDthebibliography\thebibliography
\renewcommand\thebibliography[1]{
  \OLDthebibliography{#1}
  \setlength{\itemsep}{.4ex}
}

\newcommand{\Z}{{\mathbb{Z}}}

\newcommand{\F}{{\mathbb{F}}}
\newcommand{\R}{\mathbb{R}}
\newcommand{\E}{\mathop{\mathbb{E}}}
\newcommand{\Pdist}{\mathbb{P}}

\newcommand{\Bs}{{\mathcal{B}}}
\newcommand{\Cs}{{\mathcal{C}}}
\newcommand{\Ds}{{\mathcal{D}}}

\newcommand{\Os}{{\mathcal{O}}}
\newcommand{\Ps}{{\mathcal{P}}}

\newcommand{\Rs}{{\mathcal{R}}}
\newcommand{\Ss}{{\mathcal{S}}}

\newcommand{\Us}{{\mathcal{U}}}

\newcommand{\Xs}{{\mathcal{X}}}

\newcommand{\Am}{{\mathbf{A}}}

\newcommand{\Sm}{{\mathbf{S}}}

\newcommand{\Xm}{{\mathbf{X}}}

\newcommand{\cv}{{\Vec{c}}}

\newcommand{\ev}{{\Vec{e}}}

\newcommand{\sv}{{\Vec{s}}}

\newcommand{\vv}{{\Vec{v}}}

\newcommand{\xv}{{\Vec{x}}}
\newcommand{\yv}{{\Vec{y}}}
\newcommand{\zv}{{\Vec{z}}}

\begin{document}
\title{\LARGE Exponential quantum space advantage in random data streams}

\author[1]{Adam Bouland\thanks{ \texttt{abouland@stanford.edu}}}
\author[1]{Matthew Ding\thanks{ \texttt{mcding@stanford.edu}}}
\author[2]{Siddhartha Jain\thanks{ \texttt{sidjain@utexas.edu}}}

\affil[1]{Department of Computer Science, Stanford University}
\affil[2]{Department of Computer Science, The University of Texas at Austin}

\date{November 2, 2026}
\maketitle

\begin{abstract}
    We show two unconditional quantum space advantages in the random-order streaming model.
    First, we show the Yamakawa--Zhandry Code Intersection problem admits exponential quantum advantage in the streaming model when its inputs are streamed in random order.
    This means quantum computers exhibit exponential space advantage even when simply receiving $(x,f(x))$ pairs for a uniformly random function $f$ in a uniformly random order.
    Notably, our classical space lower bound does not follow from a communication lower bound. Instead, we utilize density-restoring partitions from communication complexity directly with potential-function analysis in the style used to prove memory-sample tradeoffs for learning. The resulting argument is distinct from both lines of work.
    
    Second, we use our framework to show quantum space advantage for the Optimal Polynomial Intersection (OPI) problem in certain regimes via a streaming version of the Decoded Quantum Interferometry algorithm.
    In particular, we show that for degree $d$ and $n$ evaluation points, attaining $1/2 + \Omega(\sqrt{d/n})$ fraction of satisfied OPI constraints via streaming requires $\Omega(n)$ classical bits of memory but only $O(d\log n)$ qubits.
    This yields provable quantum advantage in a ``low-rate'' regime when the number of evaluation points is much larger than the degree, with a space advantage that can be as large as exponential in certain parameter settings.
\end{abstract}

\newpage

\section{Introduction}

Quantum computation offers numerous forms of computational advantage, reducing the time \cite{Shor_1997, HHL_2009}, space \cite{watrous1999space,ta2013inverting,kretschmer2026unconditional}, or communication \cite{Raz99exponential, Yossef04,KR11} needed to solve problems.
The streaming model---where one accesses the input one component at a time, and the goal is to solve the problem using far less space than the input length---is a natural testing ground for quantum advantage for processing classical data.
In streaming algorithms, a well-studied model is the \emph{random-order model}, where the input is revealed in a uniformly random order through the data stream \cite{munro1980selection,KKS14,KKS15,assadi2021random,chiplunkar2022factorial,hashemi2024weighted,alexandru2026unit}.

While several prior works have studied quantum advantage in streaming algorithms, e.g., \cite{le2006exponential,gavinsky2009exponential,montanaro2016frequency,kallaugher2024exponential,wang2026exponential,feng2026entropy}, so far any exponential quantum space advantage has only been proven with more structured or worst-case streaming models.
For example, Kallaugher, Parekh, and Voronova showed that quantum algorithms can have exponential space advantage for the maximum directed cut problem when the stream is adversarially ordered \cite{kallaugher2024exponential}, however this advantage no longer holds when the stream is randomly ordered \cite{saxena2023streaming}.
Recently, Zhao et al.\ introduced a correlated data stream model which admits exponential quantum space advantage for many natural machine learning tasks \cite{zhao2026exponential}, but where the data stream is not i.i.d.\ or uniformly ordered.\footnote{For example, their classical lower bound lifted from query separations is proven when the inputs are sampled randomly in blocks with a particular structure.}
There have also been a number of works lifting quantum communication advantages to streaming advantages, starting with work of Le Gall \cite{le2006exponential}; in general any communication lower bound can be turned into a streaming lower bound if Alice's input is streamed before Bob's, leading to quantum space separations, e.g., \cite{gavinsky2009exponential}. 
Recently G\"{o}\"{o}s, Gur, Jain, and Li \cite{goos2025quantum} studied a variant of the Code Intersection problem of Yamakawa and Zhandry \cite{Yamakawa2022} and showed that it admits an exponential quantum advantage in communication complexity, which can likewise be lifted to a form of streaming advantage in a worst-case ordering. 
There have also been works demonstrating polynomial space improvements for quantum streaming e.g. \cite{kallaugher2022quantum}.

In summary, prior works showing exponential space advantage have addressed structured data streams.
This leads to the natural question, which to the best of our knowledge remains open:
\begin{center}
    \textbf{\emph{Can quantum algorithms achieve exponential advantage for space complexity in the streaming model when the inputs are presented in a uniformly random order?}}
\end{center}
Random-order streaming can be thought of as relaxing the standard ``doubly worst-case'' streaming model (worst-case input \emph{and} order). Proving lower bounds for random ordering reveals average-case complexity and gives stronger evidence for a problem's hardness compared to a worst-case ordering
\cite{CJP2008selection,guha2009stream, CCM2016robust}.

\subsection{Our results}
We answer the above question in the affirmative. In particular, we show that the Code Intersection problem of Yamakawa and Zhandry \cite{Yamakawa2022} admits an exponential quantum space advantage when its inputs are randomly ordered. Recall that in the Code Intersection problem, the goal is to output a codeword of an error-correcting code where each entry is constrained by a random hash function. We show the following result:

\begin{theorem}[informal version of \cref{thm:yz-lower-bound,thm:YZ-algorithm}]\label{thm:yz-lower-bound-informal}
    Consider the streaming Code Intersection problem where the random oracle $H:\{0,1\}^{\poly(n)} \rightarrow \{0,1\}$ is streamed one entry at a time in a uniformly random order. While the Yamakawa--Zhandry algorithm solves this with $\poly(n)$ qubits of space, any classical streaming algorithm must use at least $\exp({n^{\Omega(1)}})$ bits of space.
\end{theorem}

In particular, this means that quantum computers have exponential advantage even when processing completely unstructured streams, i.e., when streaming entries of a uniformly random function $f$ one at a time through $(x,f(x))$ pairs in a uniformly random order. Thus our result can be interpreted as showing ``quantum advantage without structure'' in the streaming model in terms of both the input values as well as the stream ordering.

We then extend our techniques to the Optimal Polynomial Intersection (OPI) problem, introduced in the study of the DQI algorithm \cite{jordan2025dqi}. Here the goal is to find a low-degree polynomial over a finite field intersecting as many specified sets as possible. Fascinatingly, the DQI algorithm has produced better approximation ratios than the best known classical algorithm in certain parameter regimes, raising the possibility of quantum advantage for this optimization problem~\cite{khattar2025optimized}.

We show the OPI problem unconditionally yields quantum space advantage in the streaming model when the inputs are randomly ordered, but only in certain parameter regimes. In particular, the OPI problem is parameterized by the degree upper bound $d$ of the polynomial, the number of evaluation points $n$ of the polynomial, and the distribution over sets to be intersected at each point. 
We show that the DQI algorithm exhibits a space advantage for OPI of the following form:

\begin{theorem}[informal version of \cref{thm:opi-lower-bound,thm:DQI-algorithm}]\label{thm:opi-lower-bound-informal}
    Consider the streaming OPI problem where the inputs appear in a uniformly random order. While the DQI algorithm achieves a value of $\lceil n/2 + \sqrt{dn}/8\rceil$ for this problem using $O(d \log n)$ qubits of space, any classical streaming algorithm achieving this value must use $\Omega(n)$ classical bits of space.
\end{theorem}

Thus DQI achieves a streaming advantage when the number of evaluation points $n$ is much more than the degree $d$.\footnote{We note that the full theorem statement requires $d=\omega(\log n)$ and $d=o(n/\log n)$ to achieve the stated bounds.}
This advantage can be exponential in certain asymptotic regimes: for example, if $d=\Theta(\log^2n)$, our theorem results in an $O(\log^3 n)$ quantum versus $\Omega(n)$ classical space advantage, though we note this parameter setting is non-standard in prior DQI optimization works.
Our theorem also gives large polynomial advantages in other regimes, e.g., when $d\sim n^{0.01}$. 
Interestingly, the ``low-rate'' setting where $d\ll n$ is one setting where there is empirical evidence for quantum advantage with DQI in the optimization setting \cite{khattar2025optimized}.\footnote{For example, one finite-sized instance considered in \cite{khattar2025optimized} sets $n=4095, d=70$.}
As streaming can be seen as a simplified ``testbed'' for demonstrating quantum advantage similar to query or communication complexity, our result can be interpreted as indirect evidence that DQI might offer advantage for optimization.
However, we emphasize that our lower bounds only apply to streaming models, and not to the original OPI problem.

\subsection{Proof sketch}
\paragraph{Quantum streaming algorithms.} 
To show quantum advantage for these problems, we first need to show that they admit space-efficient quantum streaming algorithms. Our algorithms rely on the fact that both the Yamakawa--Zhandry algorithm and the DQI algorithm are completely nonadaptive, i.e., they can be implemented in two phases where the first phase depends on the input and the second phase is independent of the input. In this case, the first phase requires only input-dependent commuting gates, which can be implemented while the data is streaming with a single pass.\footnote{Note that one can also use the quantum oracle sketching algorithm of Zhao et al.~\cite{zhao2026exponential}, but this machinery is not required here as a much simpler streaming algorithm suffices.} We particularly use a more space-efficient version of DQI \cite{khattar2025optimized}. We give simple proofs of the correctness of the streaming algorithms in \cref{sec:algorithms}.

\paragraph{Prior classical lower bounds do not suffice.}
The main challenge is now to show classical lower bounds for these streaming problems.
This does not follow from prior works lower bounding the communication complexity of these problems, essentially because the random ordering of the stream breaks up the communication structure of the problem. 
For example, the lower bound of ~\cite{goos2025quantum} for the communication complexity of Code Intersection only works when the input is ordered in a certain way, namely when one first sees the constraints on the first half of the codeword, and then sees the constraints on the second half of the codeword.
More generally there is a barrier for using communication complexity for our lower bound, because there is actually a \emph{zero-communication} protocol for Code Intersection when the input is randomly distributed between two parties.\footnote{This is because even if one has only seen half of the input stream, with high probability there exists a valid codeword whose coordinates all map to 1 within the inputs observed so far. Thus the communication complexity of this problem is trivial across any partition of a random-order stream into a prefix and suffix.}
Another approach is to use general lower bounds for Read-Once Branching Programs (ROBPs) which imply streaming space lower bounds. Here, there is a line of work~\cite{raz2018fast,Raz17,GargRazTal18,Moshkovitz18} giving generic lower bounds, but these techniques are specialized for problems with a \emph{unique} planted solution. Both YZ and OPI are search problems with exponentially many valid solutions, so their technique does not work as a black-box for us. 
  
\emph{Our approach:} To achieve our lower bound, we combine the techniques from the recent communication lower bound for Code Intersection \cite{goos2025quantum} with the technique for proving ROBP lower bounds for space-bounded learning problems ~\cite{raz2018fast,Raz17,GargRazTal18,Moshkovitz18}. This fusion of two important tools for studying communication and space complexity may be of independent interest.

\paragraph{Density-restoring partitions.} At a high level, our lower bound is obtained by showing that the output of the Yamakawa--Zhandry problem meaningfully depends on a large fraction of the input.
We formalize this in our lower bound by quantifying the \emph{uncertainty} in a marginal distribution of inputs conditioned on the current state of the classical algorithm via its min-entropy. An important tool here is density-restoring partitions from query-to-communication lifting literature~\cite{GPW20,CFKMP21,goos2025quantum}.
This is a tool to `concentrate' the certainty in a given input onto  a small marginal distribution which is fixed, leaving all the other marginal distributions with near-maximal uncertainty. This makes the analysis closer to proving query lower bounds.
The work of \cite{goos2025quantum} decomposes the distribution after conditioning on the transcript of the communication protocol into `subcube-like' parts, which is a generalization of query algorithms. They then show that one can show a lower bound for subcube-like protocols using an argument similar to the query lower bound of~\cite{Yamakawa2022}.

However, we cannot reduce to subcubes in our setting. This is because their proof relies on the property that `if Alice is the first player to fix 80\% of her input, Bob has less than 80\% of his input fixed.' This relies on the input being split coordinate-wise, which is no longer true once we stream the oracle in a random order. 

\paragraph{Read-once branching programs.} To fix this problem, we  instead directly lower bound the width of ROBPs for solving the search problem.
Given these pre-processed distributions, which are either `fixed' (have no uncertainty) or `dense' (have near-maximal uncertainty), we use a simple potential function to measure the progress made by a classical algorithm modeled as an ROBP. This is similar to the approach of \cite{GargRazTal18}, who also use a convex potential function to track progress, whose results were further inspired by \cite{Raz17,raz2018fast,Moshkovitz18}. The details of our analysis differ significantly and are specialized to the setting of Code Intersection. Recall that in the Code Intersection problem we want to find a codeword which is mapped to $1^n$ by a random hash function, which is the input.
Intuitively, our potential function is just the conditional probability that a codeword is mapped to $1^n$, summed over all the codewords which have already crossed a fixed threshold of mapped coordinates.
Formally, this requires a couple of technical changes as described in \cref{sec:YZ-lower-bound}.
We crucially use the list-recoverability properties of the folded Reed-Solomon code to prove the lower bound, similar to the previous lower bounds for the Code Intersection problem~\cite{Yamakawa2022,goos2025quantum,bostanci2026advice}.

\paragraph{Fine-grained analysis for OPI.}
Extending these results to OPI requires more careful examination.
For OPI, we have a constraint satisfaction problem (CSP) and we want to say that the classical algorithm cannot achieve a high value with small space. To handle this, we modify our potential function to account for this allowed approximation. This is actually a simple modification of our threshold potential function, and it fits our framework for proving lower bounds for Code Intersection. The second place where our lower bound for OPI differs is that we forego the use of density-restoring partitions to bound the entropy in marginal distributions, and explicitly use the fact that two distinct polynomials of degree less than $d$ can agree on at most $d-1$ points.
This additional technical work leads to a more fine-grained lower bound for the corresponding ROBPs, see \cref{sec:OPI-lower-bound} for details.

\subsection{Acknowledgments}
We thank Avishay Tal for helpful discussions, and Stephen Jordan and Haimeng Zhao for their suggestions on the manuscript. 
A.B.\ was partially supported by the National Science Foundation under award No.\ 2016245 and 2440805, by the Air Force Office of Scientific Research under grant agreement FA9550-21-1-0392 and FA9550-24-1-0089, and by the DOE Office of Science under grant agreement DE-SC0025934. 
M.D.\ was supported in part by the National Science Foundation Graduate Research Fellowship Program under Grant No.\ DGE-2146755. 
S.J.\ was supported by the Amazon AI fellowship.

\paragraph{Concurrent work.}
We recently became aware of two independent and complementary works: (1) Wu and Yang~\cite{WuYang2026} showed streaming advantage for DQI in an incomparable streaming setting (worst-case order, multipass), and (2) Voronova~\cite{voronova2026random} showed exponential quantum space advantage for an ``artificial'' problem based on Hidden Matching in an incomparable random-order streaming setting (repeated inputs, online outputs).

\paragraph{AI disclosure.}
The authors proposed the central questions of this work and identified the barriers to using existing proof techniques.
ChatGPT 6 Astra provided the first proof sketch of the classical streaming lower bound, which the authors then formalized, simplified, and verified. 
The authors contributed all writing and exposition, while ChatGPT and Claude assisted with proofreading and minor suggestions.
The authors take responsibility for the originality and correctness of all proofs.

\section{Preliminaries}
$[n]$ denotes the set $\{1, \ldots, n\}$. For finite sets $S, T$, we denote $S^T$ as the collection of all functions $f: T \rightarrow S$.  For a string $\sv\in \bits^S$, let $|\sv|$ denote the Hamming weight of $\sv$. For a random function $\H: \Us \rightarrow \bits$ and a subset of the domain $S\subseteq \Us$, we write $\H_{S}$ to denote the restriction of $\H$ onto the subset $S$, and write $\H_{S} \equiv 1$ if $\H(u)=1$ for all $u\in S$. Vectors are denoted $\vv$. $\log$ denotes base-2 logarithm and $\ln$ denotes the natural logarithm. Let $[x]^+ \triangleq \max(x, 0)$ for $x\in\R$. We use the convention for binomial coefficients that $\binom{n}{k}=0$ for any $k\notin \{0, 1, \ldots, n\}$.

\subsection{Probability}
Random variables $\Xm$ are denoted in bold. Let $\mathrm{supp}(\Xm)$ denote the support of a random variable $\Xm$. We use $\mathds{1}[E]$ to denote the indicator random variable for the event $E$. 

\begin{proposition}[Hoeffding's inequality]\label{prop:hoeffding}
    If $\Xm_1, \Xm_2, \ldots, \Xm_m$ are independent random variables where $\forall i: a_i \leq \Xm_i \leq b_i$ and $\Sm_m = \Xm_1 + \Xm_2 + \cdots + \Xm_m$, then for $\delta\geq 0$,
    \begin{equation}
        \Pr(\Sm_m - \E \Sm_m \geq \delta) \leq e^{-\left(\frac{2\delta^2}{\sum_{i}(b_i-a_i)^2}\right)}.
    \end{equation}
\end{proposition}

\subsection{Streaming model and problem definition}
Let $\Sigma$ denote an alphabet and $n$ denote the block length of a code. Define the \emph{address set} $\Us \triangleq [n]\times\Sigma$, and set $M\triangleq |\Us| = n|\Sigma|$. A hidden input $(\H, \bm{\pi})$ is sampled from a predefined distribution, where $\H: \Us\rightarrow\bits$ and $\bm{\pi}: [M] \rightarrow \Us$, and the input stream is then given as
\begin{equation}
    \left(\bm{\pi}(1), \H(\bm{\pi}(1))\right), \left(\bm{\pi}(2), \H(\bm{\pi}(2))\right),\ldots, \left(\bm{\pi}(M), \H(\bm{\pi}(M))\right).
\end{equation}
We will often treat $\H: [n]\times \Sigma \rightarrow \bits$ as $n$ separate functions and use the notation $\H_i(\cdot)\triangleq \H(i, \cdot): \Sigma\rightarrow\bits$ for convenience.

Now we formally define the $\CI$ and $\OPI$ problems in the streaming model\footnote{We define our problems for any arbitrary stream ordering. Our algorithms succeed with high probability given any fixed worst-case stream ordering, while our lower bounds are with high probability over a uniformly random ordering. Note that the random-order lower bound is a \emph{stronger} statement compared to a worst-case lower bound.}:
    
\begin{problem}[$\CI$ ($\CIshort$)]\label{problem:YZ-streaming}
    Let $\Cs \subseteq (\F_q^m)^n$ be an error-correcting code. Sample a uniformly random oracle $\H: [n]\times \F_q^m \rightarrow \bits$. The stream consists of the elements from set 
    \begin{equation}
        \Ss_{\CIshort} \triangleq \{(i,x, \H_i(x)): i\in[n], x\in\F_q^m\}
    \end{equation}
    in some ordering $\bm{\pi}$. Given a single pass over the stream, the goal is to output a codeword $\cv=(c_1,\dots,c_n)\in \Cs$ satisfying
    \begin{equation}\label{eqn:YZ-condition}
        \H_i(c_i)=1, \quad \forall i\in[n].
    \end{equation}
\end{problem}

\begin{problem}[$\OPI$ ($\OPIshort$)]\label{problem:OPI-streaming}
    Let $q$ be a power of two, $n=q-1$, and $1\leq d\leq n$. Fix distinct evaluation points $\alpha_1,\ldots,\alpha_n\in\F_q^\ast$, and write $\Ps_d\triangleq\{f\in\F_q[X]:\deg f<d\}$. Sample each target set $F_1, F_2, \ldots, F_n \subseteq \F_q$ independently by including each element of $\F_q$ using an i.i.d Bernoulli with probability $1/2$. The stream consists of the elements from set 
    \begin{equation}
        \Ss_{\OPIshort} \triangleq \{(i,x, \mathds{1}[x\in F_i]): i\in[n], x\in\F_q\}
    \end{equation}
    in some ordering $\bm{\pi}$. Given a single pass over the stream and a fixed value $\tau \in [n]$, the goal is to output a polynomial $f\in\Ps_d$ such that
    \begin{equation}\label{eqn:OPI-value}
        \OPIval(f)\triangleq \sum_{i=1}^n \mathds{1}[f(\alpha_i)\in F_i] \geq \tau.
    \end{equation}
\end{problem}

Both of these problems can be unified via a search problem of a generalized Code Intersection relation in the streaming model:
\begin{definition}[Generalized search problem]\label{def:code-relation}
    Fix some code $\Cs\subseteq\Sigma^n$. Given a uniformly random sampled function $\H: [n]\times\Sigma\rightarrow\bits$, define the generalized Code Intersection relation $\Rs_{\Cs, \H}\subseteq \bitsn\times\Sigma^n$ as
    \begin{equation}
        \Rs_{\Cs, \H} \triangleq \{(\yv, \cv) \in \bitsn\times\Sigma^n: (\cv\in\Cs) \wedge (\forall i\in[n], \H_i(c_i)=y_i)\}.
    \end{equation}
    The stream consists of the elements from set $\{(i,x, \H_i(x)): i\in[n], x\in\Sigma\}$. 
\end{definition}
\begin{itemize}
    \item \cref{problem:YZ-streaming}: Set $\Sigma \gets \F_q^m$. The goal of $\CIshort$ is equivalent to outputting any codeword $\cv\in\Cs$ such that $(\Vec{1}, \cv)\in  \Rs_{\Cs, \H}$. 

    \item \cref{problem:OPI-streaming}: Set $\Sigma \gets \F_q$. For polynomial $f\in\Ps_d$, let $\cv_f\triangleq(f(\alpha_1),\ldots,f(\alpha_n))$ denote its evaluation vector. We define $\Cs$ as the Reed--Solomon code
    \begin{equation}
        \Cs\triangleq\RS_q(n,d)=\{\cv_f: f\in\Ps_d\}\subseteq\F_q^n.
    \end{equation}
    We can equivalently define $F_i \triangleq \{x\in \F_q: \H_i(x)=1\}$. The goal of $\OPIshort$ is equivalent to outputting any codeword $\cv\in\Cs$ such that $(\yv, \cv)\in \Rs_{\Cs, \H}$ and $|\yv| \geq \tau$.
\end{itemize}

For the remainder of the paper we will primarily use the codeword search problem interpretation for $\OPI$, which is equivalent to \cref{problem:OPI-streaming} up to an isomorphism.

\subsection{Branching programs}
We assume without loss of generality that our streaming algorithm is deterministic, as we can fix the randomness of any randomized streaming algorithm. We model computation of streaming algorithms using a branching program. Branching programs represent any possible unbounded computation using bounded memory and sample complexity. This model has been used successfully to prove memory-sample tradeoffs for a large class of learning problems \cite{raz2018fast, moshkovitz17a, GargRazTal18}.

\begin{definition}[Branching program]
    A \emph{branching program} $\Bs$ with length $M$ and width $w$ is a layered directed graph with layers $0, 1, \ldots, M$ (which we sometimes refer to as time $t$) and vertices $V \subseteq \{0,1,\ldots, M\}\times[w]$. Layer 0 has a single vertex, called the start vertex. Layer $M$ contains only leaf vertices with out-degree 0, each corresponding to a specific output of the branching program. All other vertices in layers $0\leq t < M$ have out-degree $2M$ with all edges pointing towards some vertex in layer $t+1$. Each one of these $2M$ edges is indexed by a unique element in the set $\Us \times \bits$. A fixed input stream defines a unique computation path which traverses exactly one vertex in each layer. The goal of the branching program is to compute some predefined relation $\Rs_{\H}$ which will be independent of the input order.
\end{definition}

\paragraph{Analysis notation.}
We introduce additional notation for our branching program analysis. $\Us\triangleq [n]\times\Sigma$ denotes the set of all addresses and $\Us_t \triangleq \binom{\Us}{t}$ denotes the set of all possible subsets of addresses which are revealed up to time step $t$. Let $V_t$ denote the vertices in layer $t$, and let $\Gamma_t$ denote the edges from layer $t-1$ to layer $t$. For a vertex $v$ in $\Bs$, let $\Gamma_{\text{in}}(v)$ denote the incoming edges into vertex $v$ and $\Gamma_{\text{out}}(v)$ denote the outgoing edges from vertex $v$.

All probabilities and expectations are implicitly over the uniform distribution of all random oracles $\H: \Us \rightarrow \bits$ and the uniform over all bijections $\bm{\pi}: [M]\rightarrow \Us$, unless denoted otherwise. For fixed vertex $v\in V_t$ in layer $0\leq t \leq M$, edge $e\in \Gamma_{t}$ incoming into layer $1\leq t\leq M$, and size-$t$ subset $A\subseteq \Us$, we define the events
\begin{equation}
    E_{v,A} \triangleq \{\bm{A_t} = A, \bm{v_t} = v\}, \qquad E_{e,A} \triangleq \{\bm{A_t} = A, \bm{e_t} = e\},
\end{equation}
and similarly for $E_v$, $E_e$, and $E_A$. We let $\bm{v_t}$ be the vertex that is reached in $V_t$ and $\bm{e_t}$ be the edge that is reached in $\Gamma_t$. Let $\bm{A_t} \triangleq \{\bm{\pi}(1), \bm{\pi}(2),\ldots, \bm{\pi}(t)\}$ denote the addresses appearing in the first $t$ entries of the input stream.

For any fixed time $0 \leq t \leq M$, we use the following notation for distributions conditional on $e\in \Gamma_t, v\in V_t, A\in \Us_t$:
\begin{equation}
    \Pdist_{\H\mid v,A}(\cdot) \triangleq \begin{cases} 
      \Pr_{\H,\bm{\pi}}\left(\cdot  \mid E_{v,A}\right) & \text{if } \Pr_{\H, \bm{\pi}}(E_{v,A}) > 0 \\
      0 & \text{if } \Pr_{\H, \bm{\pi}}(E_{v,A}) = 0
   \end{cases}
\end{equation}
and define $\Pdist_{\H\mid e,A}$ likewise. For $t=M$, we necessarily have $A = \Us$, and we will use $\Pdist_{\H\mid v}$ as shorthand for $\Pdist_{\H\mid v,\Us}$. For a vertex $v\in V_M$, let $\cv_v$ denote its associated output (we can assume without loss of generality that $\cv_v$ is a valid codeword in $\Cs$).

Define $T_{\cv} \triangleq \{(i, c_i): i\in[n]\} \subseteq \Us$ which defines the set of all stream elements corresponding to codeword $\cv$. Given a set of observed addresses $A\subseteq \Us$, i.e., the addresses which have already appeared in the input stream at a given point in time, the set $T_{\cv} \cap A$ represents all observed addresses corresponding to codeword $\cv$. Let $A^C \triangleq \Us \setminus A$ denote the set of addresses which have yet to appear in the input stream. The set $T_{\cv} \cap A^C$ is then the set of addresses corresponding to $\cv$ which have not yet appeared in the stream. Define $k_{\cv}\triangleq |T_{\cv} \cap A|$ which denotes the number of the observed addresses corresponding to $\cv$, with the dependence on $A$ implied.

Given a function $\H: \Us\rightarrow\bits$ and a subset $S \subseteq \Us$, we define the value function:
\begin{equation}
    \val_{\H}(\cv; S) \triangleq |\{(i, c_i) \in T_{\cv}\cap S: \H_i(c_i)=1\}|.
\end{equation}
The $\CI$ condition in \cref{eqn:YZ-condition} is equivalently given as $\val_{\H}(\cv; \Us)=n$, and the $\OPIshort$ condition in \cref{eqn:OPI-value} is equivalently given as $\val_{\H}(\cv; \Us) \geq \tau$. When $S=\Us$, we will sometimes omit the parameter $S$ and write $\val_{\H}(\cv)$ instead.

\subsection{List-recoverable codes}
List recoverability is the main property we use with Code Intersection to bound the number of codewords which become ``too structured'', i.e., having too many coordinates be observed along with high posterior probability of all of them mapping to 1. We define list-recoverable codes with respect to total list size (which is implied by list-recoverability with respect to individual list sizes):
\begin{definition}[List recoverability]
    A code $\Cs \subseteq \Sigma^n$ is called $(\rho, \ell, L)$-list-recoverable if for any sets $S_1, \ldots, S_n$ such that $\sum_{i=1}^n |S_i| \leq \ell$, we have
    \begin{equation}
        \left|\{\cv\in\Cs: |T_{\cv}\cap S|\geq \rho n\}\right|\leq L,
    \end{equation}
    where $S\triangleq \bigcup_{i\in[n]} \{i\} \times S_i$.
\end{definition}

\begin{proposition}[\cite{Yamakawa2022}, Lemma 4.2]\label{prop:YZ-list-recovery}
    Fix any constants $0<c<c'<1$. There exists an explicit family of folded Reed--Solomon codes $\{\Cs_n\}$ over field $\F_q$, where $\Cs_n \subseteq \Sigma^n$ with block length $n$, and folding parameter $m$, such that
    \begin{equation}
        \Sigma = \F_q^m, \qquad q = \Theta(n^2),\qquad m=\Theta(n),\qquad \log|\Sigma|=\Theta(n\log n),\qquad \log|\Cs_n|=O(n^2\log n).
    \end{equation}
    $\Cs_n$ also satisfies the following properties:
    \begin{enumerate}
        \item (List-recoverability, used in classical lower bound \cref{thm:yz-lower-bound}) $\Cs_n$ is $(\rho, \ell, L)$-list-recoverable, where $\rho \in (0,1)$ is a fixed constant, $\ell = 2^{\Theta(n^c)}$, and $L=2^{\Tilde{\Os}(n^{c'})}\leq 2^{o(n)}$.

        \item (Dual decoding, used in quantum algorithm \cref{thm:YZ-algorithm})
        There exists a polynomial-time deterministic decoding algorithm $\decode$ for the dual (over $\F_q$) code $\Cs_n^\bot$ satisfying
        \begin{equation}
        \Pr_{\ev\sim\Ds^n}\left(\forall\xv\in\Cs_n^\bot:\decode(\xv + \ev)=\xv\right) \geq 1-2^{-\Omega(n)},
        \end{equation}
        where $\Ds$ is the distribution over $\Sigma$ which is $\Vec{0}$ with probability $1/2$ and uniform on all other symbols in $\Sigma\setminus\{\Vec{0}\}$.
    \end{enumerate}
\end{proposition}

\subsection{Min-entropy and density-restoring partitions}
The density-restoring partition is a technique from communication complexity based on the ``structure versus randomness'' paradigm. It represents a distribution with high min-entropy as a convex combination of distributions which are `fixed' on some marginal (zero entropy) and have near-maximal min-entropy on every other marginal which does not intersect with the fixed part.

\begin{definition}[Min-entropy]
    The min-entropy of a random variable $\Xm$, denoted $H_{\infty}(\Xm)$, is defined as
    \begin{equation}
        H_{\infty}(\Xm) \triangleq \min_{x\in \mathrm{supp}(\Xm)} \log\left(\frac{1}{\Pr(\Xm=x)}\right).
    \end{equation}
\end{definition}

We give a useful result on min-entropy applied to our branching program notation. It is the only result which depends directly on the width bound $w=2^S$ of the branching program.
\begin{proposition}[Lemma 2.2 of \cite{dodis2008fuzzy}]\label{prop:min-entropy-deficiency}
Fix $0\leq t\leq M$ and $A$ with $|A|=t$. Except for vertices $v\in V_t$ of total probability mass at most $\varepsilon>0$ conditioned on $\bm{A_t}=A$, we have
\begin{equation}
    H_\infty(\H_A\mid E_{v,A})\geq t-S-\log (1/\varepsilon).
\end{equation}
\end{proposition}

Next, we formally introduce the density-restoring partition, originally due to \cite{GPW20}:
\begin{definition}[Density]
    Fix a set $J$ of coordinates and consider the random variable $\Xm$ over $\bits^J$. $\Xm$ is called $\gamma$-\emph{dense} if for all subsets $I \subseteq J$,
    \begin{equation}
        H_{\infty}(\Xm_I) \geq \gamma |I|.
    \end{equation}
    In other words $\Xm_I$ has $\gamma$ fraction of the maximum min-entropy achieved by any distribution over $\bits^I$.
\end{definition}
\begin{lemma}[Density-restoring partition, Lemma 3.7 of \cite{CFKMP21}]\label{lemma:density-restoration}
    Fix $0<\gamma \leq 1$, coordinate set $J$, and let $\Xm$ be a random variable with support defined on $\Xs \subseteq \bits^J$. There exists a partition
    \begin{equation}\label{eq:decomposition}
        \Xs = \Xs^{(1)} \sqcup \Xs^{(2)}\sqcup\cdots\sqcup \Xs^{(r)}
    \end{equation}
    where every $\Xs^{(j)}$ is associated with a set $I_j \subseteq J$ and an assignment $\alpha_j \in \bits^{I_j}$ such that for all $j \in [r]$,
    \begin{enumerate}
        \item The marginal distribution $\Xm_{I_j} \mid \Xm \in \Xs^{(j)}$ is fixed to $\alpha_j$.
        \item The marginal distribution $\Xm_{J\setminus I_j}\mid \Xm\in \Xs^{(j)}$ is $\gamma$-dense.
        \item The marginal distribution $\Xm_{J\setminus I_j}\mid \Xm\in \Xs^{(j)}$ has min-entropy
        \begin{equation}
            H_{\infty}(\Xm_{J \setminus I_j} \mid \Xm\in \Xs^{(j)}) \geq H_{\infty}(\Xm) - \gamma |I_j|  + \delta_j,
        \end{equation}
        where $\delta_j \triangleq \log\Pr\left(\Xm \in \cup_{i\geq j} \Xs^{(i)}\right)$.
    \end{enumerate}
\end{lemma}
\begin{corollary}\label{corollary:dense-excess-partition}
    Fix $0 <\gamma, \eta < 1$, $D\geq 0$, set $J$. Let $\Xm$ be a random variable satisfying $H_{\infty}(\Xm)\geq |J|-D$. Apply the density-restoring partition of \cref{lemma:density-restoration}: the collection $\Xs^{(1)}, \Xs^{(2)},\ldots, \Xs^{(r)}$ from \cref{eq:decomposition} can be partitioned into two groups with the following properties:
    \begin{enumerate}
        \item \textbf{Core group} $\Xs^{\mathrm{core}}$: Every $\Xs^{(j)}$ in the group satisfies
        \begin{equation}
            |I_j| \leq \frac{D + \log (1/\eta)}{1-\gamma}.
        \end{equation}

        \item \textbf{Excess group} $\Xs^{\mathrm{excess}}$:
        \begin{equation}
            \Pr(\Xm\in \Xs^{\mathrm{excess}}) < \eta.
        \end{equation}
    \end{enumerate}
\end{corollary}
\begin{proof}
    For each part $\Xs^{(j)}$ in the density-restoring partition $\Xs = \Xs^{(1)} \sqcup \Xs^{(2)}\sqcup \cdots \sqcup \Xs^{(r)}$, any variable with support on a subset of $\bits^{J\setminus I_j}$ has min-entropy at most $|J|-|I_j|$. Combining this with property (3) of \cref{lemma:density-restoration},
    \begin{equation}
        |J| - D - \gamma |I_j|  + \delta_j \leq H_{\infty}(\Xm_{J \setminus I_j} \mid \Xm\in \Xs^{(j)}) \leq |J|-|I_j|.
    \end{equation}
    Let $j^*$ be the smallest index $j\in [r]$ such that $\delta_j < \log(\eta)$ (if none exists, set $j^*\gets r+1$: the excess group will be empty and the probability bound for it immediately follows). The core group is the union of $\Xs^{(j)}$ for all $j < j^*$. For each $\Xs^{(j)}\subseteq \Xs^{\mathrm{core}}$, $\delta_j \geq \log(\eta)$ and therefore
    \begin{equation}
        |J| - D - \gamma |I_j|  + \log(\eta) \leq |J|-|I_j|,
    \end{equation}
    which implies
    \begin{equation}
        |I_j| \leq \frac{D + \log (1/\eta)}{1-\gamma}.
    \end{equation}
    The excess group is the union of $\Xs^{(j)}$ for all $j \geq j^*$. Since $\delta_{j^*} < \log \eta$, we have
    \begin{equation}
        \Pr(\Xm\in \Xs^{\mathrm{excess}}) = \Pr(\Xm \in \cup_{j\geq j^*} \Xs^{(j)}) < \eta
    \end{equation}
    by definition. 
\end{proof}

\section{Classical space lower bound for streaming Code Intersection}\label{sec:YZ-lower-bound}
We prove our formal classical space lower bound for Code Intersection in a random-order stream in \cref{thm:yz-lower-bound}, which is the formal version of the lower bound in \cref{thm:yz-lower-bound-informal}. A key observation of our model specifically for our analysis is that conditional on event $E_{v,A}$, the posterior distribution of the next stream element is uniform over addresses in $A^C$ and an oracle label which is an independent, fair $0/1$ bit. Additionally, conditioning on the next received stream element leaves the posterior distribution $\H_A$ unchanged.

\subsection{Potential method for Code Intersection space lower bound}
Given a codeword $\cv$, vertex $v\in V_t$ and observed address set $A$, the posterior probability that $\H$ maps $\cv$ to $\Vec{1}$ can be split into two components: the probability of the observed addresses all mapping to 1, and the probability of the unobserved addresses all mapping to 1. Since oracle mappings on the unobserved addresses are independent of the branching program state, this gives
\begin{equation}\label{eqn:YZ-conditional-prob}
    \Pdist_{\H\mid v,A}(\val_{\H}(\cv)=n) = \underbrace{
    \vphantom{\Pdist_{\H\mid v, A}}
    2^{-(n - k_{\cv})}}_{\text{unobserved}} \underbrace{\Pdist_{\H\mid v, A}(\val_{\H}(\cv; A)=k_{\cv})}_{\text{observed}},
\end{equation}
recalling that $k_{\cv}\triangleq |T_{\cv}\cap A|$. The success probability of the branching program can be computed as
\begin{equation}\label{eqn:success-eq-1}
    \Pr(\Bs\ \text{succeeds}) = \sum_{v\in V_M} \Pr(E_v)\,\Pdist_{\H\mid v}(\val_{\H}(\cv_v)=n) \leq \sum_{v\in V_M} \Pr(E_v)\, \max_{\cv\in \Cs} \Pdist_{\H\mid v}(\val_{\H}(\cv)=n).
\end{equation}

The maximum posterior probability over all codewords conditional on a memory state is generally difficult to analyze, because branching programs are able to forget information arbitrarily. We seek to construct a potential which upper bounds it while also being simpler to deal with. A natural choice would be the summation over all codewords
\begin{equation}
    \sum_{v\in V_M} \Pr(E_v) \sum_{\cv\in \Cs} \Pdist_{\H\mid v}(\val_{\H}(\cv)=n),
\end{equation}
however this is equal to the expected number of codewords mapping to $\Vec{1}$ which is an extreme overestimate. This is where the structure versus randomness paradigm comes in: we want to construct our potential so that it only sums over codewords with potentially significant posterior probability (structured information) and ignore the codewords with very low posterior probability (random-like objects). To determine which codewords have the potential for structured information, we simply count the number of oracle points seen which agree with one of the coordinates of the codeword. We also only count probability which exceeds a small cutoff value in order to avoid over-counting many different candidates which we have low confidence in.

\begin{definition}[Threshold crossing potential]\label{def:potential}
    Fix a parameter $1\leq \kappa\leq n$ which we call the threshold. For $\eta\in(0,1/2)$, we define the potential $\Phi_{\eta, \kappa}$ at time $t$ as
    \begin{equation}
        \Phi_{\eta, \kappa}(t) \triangleq \sum_{\substack{v\in V_t\\A\in \Us_t}} \Pr(E_{v,A}) \sum_{\substack{\cv\in\Cs\\k_{\cv}\geq \kappa}} 2^{-(n-k_{\cv})} \left[\Pdist_{\H\mid v, A}(\val_{\H}(\cv; A)= k_{\cv})-2\eta\right]^+.
    \end{equation}
\end{definition}
We will omit parameters $\eta, \kappa$ when they are clear from context. Eventually we will set $\kappa \gets \lceil\rho n + \log(1/\eta)/\gamma\rceil + 1$ and $\eta = 2^{-\Omega(n)}$ in \cref{thm:yz-lower-bound}, however we keep them as general parameters in the lemmas below.

\begin{lemma}[Potential bounds $\CI$ success probability]\label{lemma:success-bound}
    For any threshold $\kappa\in [n]$ and fixed value $\eta\in(0,1/2)$,
    \begin{equation}
        \Phi(0)=0 \qquad \mathrm{and}\qquad \Pr[\Bs\, \mathrm{succeeds}] \leq 2\eta + \Phi(M).
    \end{equation}
\end{lemma}
\begin{proof}
    At $t=0$, $k_{\cv} = 0 < \kappa$ for all codewords $\cv$, so the summation over codewords is empty and thus $\Phi(0)=0$.  
    
    At $t=M$, $k_{\cv} = n \geq \kappa$ for all codewords $\cv$, so the summation reduces to
    \begin{equation}\label{eqn:success-eq-2}
        \Phi(M) = \sum_{v\in V_M} \Pr(E_v)\sum_{\cv\in\Cs} [\Pdist_{\H\mid v}(\val_{\H}(\cv)=n) - 2\eta]^+.
    \end{equation}
    Combining \cref{eqn:success-eq-1,eqn:success-eq-2},
    \begin{align}
        \Pr(\Bs\ \text{succeeds}) &\leq \sum_{v\in V_M} \Pr(E_v)\left(2\eta + \max_{\cv\in \Cs} \left[\Pdist_{\H\mid v}(\val_{\H}(\cv)=n) - 2\eta\right]^+\right) \\
        &= 2\eta + \sum_{v\in V_M} \Pr(E_v) \max_{\cv\in \Cs}\left[\Pdist_{\H\mid v}(\val_{\H}(\cv)=n) - 2\eta\right]^+ \\
        &\leq 2\eta + \sum_{v\in V_M} \Pr(E_v) \sum_{\cv\in\Cs} \left[\Pdist_{\H\mid v}(\val_{\H}(\cv)=n) - 2\eta\right]^+ \\
        &= 2\eta + \Phi(M).
    \end{align}
\end{proof}

\subsection{Applying the potential to Code Intersection}
We begin by defining the intermediate edge-parameterized potential function:
\begin{equation}
    \Psi_{\eta, \kappa}(t) \triangleq \sum_{\substack{e\in \Gamma_t\\A\in \Us_t}} \Pr(E_{e,A}) \sum_{\substack{\cv\in\Cs\\k_{\cv}\geq \kappa}} 2^{-(n-k_{\cv})} \left[\Pdist_{\H\mid e, A}(\val_{\H}(\cv; A)=k_{\cv})-2\eta\right]^+.
\end{equation}
The potential change from $\Phi(t)$ to $\Phi(t+1)$ can be split into two components:
\begin{enumerate}
    \item $\Psi_{\eta, \kappa}(t+1) - \Phi_{\eta, \kappa}(t)$, which corresponds to the potential gain incurred by starting from a vertex and going ``onto'' its outgoing edge: this corresponds to the information change from receiving a new stream input. 

    \item $\Phi_{\eta, \kappa}(t+1) - \Psi_{\eta, \kappa}(t+1)$, which corresponds to the potential gain incurred by starting from an edge and moving to the vertex it points to: this corresponds to the information change from updating our memory state after receiving a new stream input.
\end{enumerate}

We first show that contribution $(2)$ is always non-positive, analogous to \cite[Claim 4.12]{GargRazTal18}:
\begin{lemma}[Convexity bound]\label{lemma:jensens-potential}
    Fix constants $\kappa\in [n], \eta\in(0,1/2)$. We have
    \begin{equation}
        \Phi_{\eta, \kappa}(t) \leq \Psi_{\eta, \kappa}(t)
    \end{equation}
    for all $1\leq t \leq M$.
\end{lemma}
\begin{proof}
    We partition set of edges $e\in\Gamma_t$ according to the vertices in $V_t$ that they point to, allowing us to expand out the sum:
    \begin{equation}
        \Psi(t) = \sum_{\substack{v\in V_t\\A\in \Us_t}} \sum_{e\in \Gamma_{\text{in}}(v)}\Pr(E_{e,A}) \sum_{\substack{\cv\in\Cs\\k_{\cv}\geq \kappa}} 2^{-(n-k_{\cv})} \left[\Pdist_{\H\mid e, A}(\val_{\H}(\cv; A)=k_{\cv})-2\eta\right]^+.
    \end{equation}
    By law of total probability for $E_{v,A}$ with non-zero probability,
    \begin{equation}
        \Pdist_{\H\mid v, A}(\val_{\H}(\cv; A)=k_{\cv}) = \sum_{e\in \Gamma_{\text{in}}(v)}\frac{\Pr(E_{e,A})}{\Pr(E_{v,A})}  \Pdist_{\H\mid e, A}(\val_{\H}(\cv; A)=k_{\cv}).
    \end{equation}
    Since $x\mapsto \max(0,x)$ is a convex function, we can apply Jensen's inequality to get
    \begin{equation}
        \Pr(E_{v,A})\left[\Pdist_{\H\mid v, A}(\val_{\H}(\cv; A)=k_{\cv})-2\eta\right]^+ \leq \sum_{e\in \Gamma_{\text{in}}(v)}\Pr(E_{e,A})  \left[\Pdist_{\H\mid e, A}(\val_{\H}(\cv; A)=k_{\cv})-2\eta\right]^+.
    \end{equation}
    Applying the summation to both sides of the inequality,
    \begin{align}
          \Phi_{\eta, \kappa}(t) &= \sum_{\substack{v\in V_t\\A\in \Us_t}} \sum_{\substack{\cv\in\Cs\\k_{\cv}\geq \kappa}} 2^{-(n-k_{\cv})} \,\Pr(E_{v,A})\left[\Pdist_{\H\mid v, A}(\val_{\H}(\cv; A)=k_{\cv})-2\eta\right]^+ 
          \\&\leq \sum_{\substack{v\in V_t\\A\in \Us_t}} \sum_{\substack{\cv\in\Cs\\k_{\cv}\geq \kappa}} 2^{-(n-k_{\cv})} \sum_{e\in \Gamma_{\text{in}}(v)}\Pr(E_{e,A})  \left[\Pdist_{\H\mid e, A}(\val_{\H}(\cv; A)=k_{\cv})-2\eta\right]^+
          \\&= \Psi_{\eta, \kappa}(t).
    \end{align}
\end{proof}

To bound contribution (1), we introduce a lemma which bounds the posterior probability of success for codewords which have both many observed coordinates as well as a sufficiently high probability of them all mapping to 1.
\begin{lemma}[List-recovery implies low posterior probability for strong candidates]
\label{lemma:list-recovery-confidence}
    Let $\Cs\subseteq\Sigma^n$ be a
    $(\rho,\ell,L)$-list-recoverable code, and let $\Xm$ be a
    random variable on $\bits^A$, where $A\subseteq\Us$.
    Fix $\gamma\in(0,1)$, $\eta\in(0,1/2)$, $D\geq 0$. Suppose that
    \begin{equation}
        H_\infty(\Xm)\geq |A|-D\quad \mathrm{and}\quad \frac{D+\log(1/\eta)}{1-\gamma}\leq\ell.
    \end{equation}
    Then
    \begin{equation}
        \sum_{\substack{\cv\in\Cs\\k_{\cv}\geq \rho n+\log(1/\eta)/\gamma}}
        \left[\Pr(X_{T_{\cv}\cap A}\equiv 1)-2\eta\right]^+ \leq L.
    \end{equation}
\end{lemma}
\begin{proof}
    To begin, we apply \cref{lemma:density-restoration,corollary:dense-excess-partition} to get the two groups $\Xs^{\mathrm{core}} = \Xs^{(1)} \sqcup \Xs^{(2)}\sqcup\cdots\sqcup \Xs^{(j^*-1)}$ and $\Xs^{\mathrm{excess}} = \Xs^{(j^*)} \sqcup \Xs^{(j^*+1)}\sqcup\cdots\sqcup \Xs^{(r)}$.
       
    For each component $\Xs^{(j)}$ in the core group ($j<j^*$), 
    \begin{equation}
        |I_j| \leq \frac{D + \log(1/\eta)}{1-\gamma}\leq \ell.
    \end{equation}
    From condition (1) of \cref{lemma:density-restoration}, $X_{I_j}|X\in\Xs^{(j)}$ is fixed to $\alpha_j \in \bits^{I_j}$.  
    Fix a codeword $\cv\in\Cs$ where $|T_{\cv} \cap A|\geq \rho n + \log(1/\eta)/\gamma$. For each core component index $j<j^*$, if $|T_{\cv} \cap I_j| < \rho n$, then $|T_{\cv} \cap (A\setminus I_j)|>\log(1/\eta)/\gamma$, and furthermore from condition (2) of \cref{lemma:density-restoration}, $\Xm_{A\setminus I_j}\mid\Xm\in\Xs^{(j)}$ is $\gamma$-dense, so
    \begin{equation}
        \Pr(\Xm_{T_{\cv} \cap (A\setminus I_j)} \equiv 1 \mid \Xm\in\Xs^{(j)}) \leq 2^{-\gamma |T_{\cv} \cap (A\setminus I_j)|} < \eta.
    \end{equation}
    Therefore
    \begin{equation}
        \Pr(\Xm_{T_{\cv} \cap A} \equiv 1 \mid\Xm\in\Xs^{(j)}) \leq \begin{cases} 
      \eta & \text{if}\ |T_{\cv} \cap I_j| < \rho n \\
      1 & \text{if}\ |T_{\cv} \cap I_j| \geq \rho n.
   \end{cases}
    \end{equation}
    Now using the law of total probability on all core indices $j$ plus the excess group,
    \begin{equation}
        \begin{aligned}
        \Pr(\Xm_{T_{\cv} \cap A} \equiv 1) &= \Pr(\Xm_{T_{\cv} \cap A} \equiv 1 \mid \Xm\in\Xs^{\text{excess}})\Pr(\Xm\in\Xs^{\text{excess}}) \\ &\quad +  \sum_{j<j^*} \Pr(\Xm_{T_{\cv} \cap A} \equiv 1\mid \Xm\in\Xs^{(j)})\Pr(\Xm\in\Xs^{(j)}).
        \end{aligned}
    \end{equation}
    The excess contribution is bounded by $\eta$, and we split the core groups by their value of $|T_{\cv} \cap I_j|$:
    \begin{equation}
        \Pr(\Xm_{T_{\cv} \cap A} \equiv 1) \leq \eta + \eta + \sum_{\substack{j<j^*\\|T_{\cv}\cap I_j|\geq\rho n}} \Pr(\Xm\in\Xs^{(j)})
    \end{equation}
    This implies
    \begin{equation}
        \left[\Pr(\Xm_{T_{\cv}\cap A}\equiv 1)-2\eta\right]^+ \leq \sum_{\substack{j<j^*\\|T_{\cv}\cap I_j|\geq\rho n}} \Pr(\Xm\in\Xs^{(j)}).
    \end{equation}
    Now we sum over all codewords $\cv\in\Cs$ where $k_{\cv}\geq \rho n+\log(1/\eta)/\gamma$:
    \begin{align}
        \sum_{\substack{\cv\in\Cs\\k_{\cv}\geq \rho n+\log(1/\eta)/\gamma}}
        \left[\Pr(\Xm_{T_{\cv}\cap A}\equiv 1)-2\eta\right]^+ &\leq \sum_{\substack{\cv\in\Cs\\k_{\cv}\geq \rho n+\log(1/\eta)/\gamma\\j<j^*:|T_{\cv}\cap I_j|\geq\rho n}}
        \Pr(\Xm\in\Xs^{(j)})\\
        &\leq \sum_{j<j^*} \Pr(\Xm\in\Xs^{(j)}) |\{\cv\in\Cs: |T_{\cv}\cap I_j|\geq\rho n\}|\\
        &\leq \sum_{j<j^*} \Pr(\Xm\in\Xs^{(j)}) L  \leq L,
    \end{align}
    where the final inequality is from list-recovery.
\end{proof}

\cref{lemma:list-recovery-confidence} allows us to bound contribution (1) of the potential growth:
\begin{lemma}[Potential function evolution]\label{lemma:potential-growth}
    Consider constants $\gamma\in(0,1), \eta\in(0,1/2)$, and define $\kappa^*\triangleq \lceil\rho n + \log(1/\eta)/\gamma\rceil + 1 \leq n$, and let $\Cs\subseteq\Sigma^n$ be a $(\rho ,\ell,L)$-list-recoverable code. If branching program $\Bs$ has width $w=2^S$, where
    \begin{equation}\label{eqn:potential-increase-assumption}
        S\leq (1-\gamma)\ell - \log|\Cs| - \log(1/\eta),
    \end{equation}
    then for every $0\leq t < M$,
    \begin{equation}
        \Psi_{\eta,\kappa^*}(t+1) - \Phi_{\eta, \kappa^*}(t) \leq (L+1) \frac{(n-\kappa^* + 1)2^{-(n-\kappa^* +1)}}{M-t}.
    \end{equation}
\end{lemma}
\begin{proof}
    We partition set of edges $e\in\Gamma_{t+1}$, but this time according to the vertices in $V_t$ that they are outgoing from; the edges $e\in\Gamma_{\text{out}}(v)$ are indexed by elements $(i,x,b)\in[n]\times\Sigma\times\bits$. 
    
    Use $A^+$ to denote $A \cup \{(i,x)\}$ and $k_{\cv}^+$ to denote $|T_{\cv}\cap A^+|$, which are both dependent on edge $e$. Since the distribution over stream elements has non-zero support exactly on $A^C$ at time $t+1$, $\Pr(E_{e,A^+})$ is non-zero only if the edge $e$ is indexed by $(i,x)\notin A$. We can rewrite the sum then as
    \begin{equation}\label{eqn:input-labeled-potential}
        \Psi(t+1) = \sum_{\substack{v\in V_t\\A\in \Us_t}} \sum_{(i,x)\in A^C}\sum_{b\in\bits}\Pr(E_{e,A^+}) \sum_{\substack{\cv\in\Cs\\ k_{\cv}^+\geq \kappa^*}} 2^{-(n-k_{\cv}^+)} \left[\Pdist_{\H\mid e, A^+}(\val_{\H}(\cv; A^+)=k_{\cv}^+)-2\eta\right]^+.
    \end{equation}
    
    Consider an edge $e$ indexed by $(i,x,b)$ which arrives at time $t+1$. We have
    \begin{equation}
        E_{e,A^+}=E_{v,A}\cap\{\bm{\pi}(t+1)=(i,x),\ \H_i(x)=b\}.
    \end{equation}
    Since $\bm{\pi}(t+1)$ is a uniformly random element of $A^C$ conditioned on $\Am_t=A, \bm{v}_t=v$ and independent of $\H$, we have
    \begin{equation}
        \Pr\left(\bm{\pi}(t+1)=(i,x) \wedge \H_i(x)=b \mid E_{v,A}\right) = \frac{1}{2}\cdot\frac{1}{M-t},\quad \forall (i,x)\in A^C.
    \end{equation}
    and therefore
    \begin{equation}
        \Pr(E_{e,A^+})= \frac{\Pr(E_{v,A})}{2(M-t)}.
    \end{equation}

    Plugging into \cref{eqn:input-labeled-potential} and rearranging the order of summations,
    \begin{equation}\label{eq:edge-labeled-summation-2}
        \Psi(t+1) = \sum_{\substack{v\in V_t\\A\in \Us_t}} \Pr(E_{v,A}) \sum_{(i,x)\in A^C} \sum_{\substack{\cv\in\Cs\\ k_{\cv}^+\geq \kappa^*}} \sum_{b\in\bits} \frac{2^{-(n-k_{\cv}^+)}}{2(M-t)} \left[\Pdist_{\H\mid e, A^+}(\val_{\H}(\cv; A^+)=k_{\cv}^+)-2\eta\right]^+.
    \end{equation}
    
    \paragraph{Individual codeword analysis.}
     Consider a fixed $v\in V_t$ and $A\in\Us_t$. Additionally, fix a codeword $\cv\in\Cs$ and abbreviate $k\triangleq k_{\cv}$ and $k^+ \triangleq k^+_{\cv}$. We have two possibilities for $(i,x)$:
    \begin{itemize}
        \item If $x\neq c_i$, then $(i,x)\notin T_{\cv}$, and therefore $k^+ = k$ and the posterior probability on the observed addresses remains unchanged:
        \begin{equation}
            \Pdist_{\H\mid e,A^+}(\val_{\H}(\cv;A^+)=k^+_{\cv}) = \Pdist_{\H\mid v,A}(\val_{\H}(\cv;A)=k), \quad \text{if}\ x\neq c_i.
        \end{equation}

        \item If $x = c_i$, then $(i,x)\in T_{\cv}$. Since $(i,x)\in A^+$ by definition, $k^+ = k+1$. Additionally, the edge $e$ determines the original vertex $v$, the address $(i,x)$, and oracle label $b$. Together $A^+, e$ also uniquely determine $A$, and thus $k$ as well. Therefore
        \begin{equation}
            \Pdist_{\H\mid e,A^+}(\val_{\H}(\cv;A^+)=k^+_{\cv}) = (\mathds{1}[b=1]) \Pdist_{\H\mid v,A}(\val_{\H}(\cv;A)=k),\quad \text{if}\ x=c_i,
        \end{equation}
        since the event $\val_{\H}(\cv;A^+) = k + 1$ occurs if and only if both $\val_{\H}(\cv;A) = k$ and $\H_{i}(x) = 1$.  
    \end{itemize}

    For a fixed $v, A$ where $\Pr(E_{v,A})>0$, and our chosen codeword $\cv$, if $k\leq \kappa^*-2$ then it does not appear in the summation over codewords for either $\Phi(t)$ or $\Psi(t+1)$.
    
    Now consider the case $k \geq \kappa^*$. If $(i,x)\in A^C \setminus T_{\cv}$, then $k^+ = k$. There are a total of $M-t-(n-k)$ of these addresses, and for each one and any $b\in\bits$, the conditional probability term is unchanged between $\Phi(t)$ and $\Psi(t+1)$. The total contribution in \cref{eq:edge-labeled-summation-2} from any fixed codeword satisfying $k\geq\kappa^*$ is
    \begin{equation}
     \left(\frac{M-t-(n-k)}{M-t}\right) 2^{-(n-k)} \left[\Pdist_{\H\mid v, A}(\val_{\H}(\cv; A)= k)-2\eta\right]^+.
    \end{equation}
    Now consider the $n-k$ addresses $(i,x)\in T_{\cv} \cap A^C$ which are associated with $\cv$. For each of these addresses $(i,x)$ the conditional probability term is unchanged between $\Phi(t)$ and $\Psi(t+1)$ when $b=1$ and is set to 0 otherwise. On the other hand, $k^+ = k+1$, so the value $2^{-(n-{k^+})}$ doubles. Since $b=1$ with probability $1/2$, the two effects cancel. Therefore the total contribution in \cref{eq:edge-labeled-summation-2} from any codeword $\cv$ satisfying $k\geq\kappa^*$ is
    \begin{equation}
     \left(\frac{n-k}{M-t}\right) 2^{-(n-k)} \left[\Pdist_{\H\mid v, A}(\val_{\H}(\cv; A)= k)-2\eta\right]^+.
    \end{equation}
    Thus the total contribution from codewords $\cv$ satisfying $k\geq \kappa^*$ is equivalently
    \begin{equation}
         2^{-(n-k)} \left[\Pdist_{\H\mid v, A}(\val_{\H}(\cv; A)= k)-2\eta\right]^+
    \end{equation}
    for both $\Phi(t)$ and $\Psi(t+1)$. Thus we get
    \begin{equation}
        \Psi(t+1) - \Phi(t) = \sum_{\substack{v\in V_t\\A\in \Us_t}} \Pr(E_{v,A}) \sum_{\substack{\cv\in\Cs\\k_{\cv} = \kappa^* - 1}} \sum_{\substack{(i,x)\in T_{\cv} \cap A^C}} \frac{2^{-(n-\kappa^* + 1)}}{M-t} \left[\Pdist_{\H\mid v, A}(\val_{\H}(\cv; A)=\kappa^* - 1)-2\eta\right]^+. 
    \end{equation}
    Notice that the change in potential is determined entirely by \emph{codewords which have $\kappa^*$ observed coordinates for the first time!} 
    
    \paragraph{Summing over codewords that cross the threshold.}
    Since $|T_{\cv} \cap A^C| = n-{k_{\cv}} = n-\kappa^*+1$, $\Psi(t+1) - \Phi(t)$ is equal to
    \begin{equation}\label{eqn:threshold-potential-sum}
         \frac{(n-\kappa^* + 1)2^{-(n-\kappa^* + 1)}}{M-t} \E_{A\in\Us_t} \sum_{v\in V_t}\Pr(\bm{v_t}=v \mid \Am_t = A) \sum_{\substack{\cv\in\Cs\\k_{\cv} = \kappa^* - 1}} \left[\Pdist_{\H\mid v, A}(\val_{\H}(\cv; A)=\kappa^* - 1)-2\eta\right]^+. 
    \end{equation}
    To bound this innermost sum, we utilize list-recoverability. For any $A\in \Us_t$, we condition on the event $\Am_t = A$. From \cref{prop:min-entropy-deficiency}, for all vertices $v$ excluding a fraction whose total probability mass is at most $\varepsilon$ conditioned on $\bm{A_t}=A$, we have $H_\infty(\H_A\mid E_{v,A})\geq t-S-\log(1/\varepsilon)$. Call the vertices where this condition holds ``good vertices'' and the $\leq \varepsilon$ fraction of vertices where the condition does not hold the ``bad vertices''. Set $\varepsilon\gets1/|\Cs|$. For each good vertex $v$, we apply \cref{lemma:list-recovery-confidence} with $\Xm \gets \H_A \mid E_{v,A}$ and $D\gets S+\log(1/\varepsilon)$, whose precondition holds due to our assumption in \cref{eqn:potential-increase-assumption}. This gives us
    \begin{equation}
        \sum_{\substack{\cv\in\Cs\\ k_{\cv} = \kappa^* - 1}} \left[\Pdist_{\H\mid v, A}(\val_{\H}(\cv; A)=\kappa^* - 1)-2\eta\right]^+ \leq \sum_{\substack{\cv\in\Cs\\k_{\cv}\geq \rho n+\log(1/\eta)/\gamma}}
        \left[\Pr(\Xm_{T_{\cv}\cap A}\equiv 1)-2\eta\right]^+ \leq L,
    \end{equation}
    which bounds the contribution of all good vertices in \cref{eqn:threshold-potential-sum}. 

    For the bad vertices, we can use the worst-case upper bound of
    \begin{equation}
        \sum_{\substack{\cv\in\Cs\\k_{\cv}\geq \rho n+\log(1/\eta)/\gamma}}
        \left[\Pr(\Xm_{T_{\cv}\cap A}\equiv 1)-2\eta\right]^+ \leq |\{\cv\in\Cs:k_{\cv}\geq \rho n+\log(1/\eta)/\gamma\}| \leq |\Cs|.
    \end{equation}
    Since the probability mass of the bad vertices conditioned on $A$ is at most $\varepsilon$, 
    \begin{equation}
        \sum_{\substack{v\in V_t\\v\ \text{bad}}}\Pr(\bm{v_t}=v \mid \bm{A_t}=A) \sum_{\substack{\cv\in\Cs\\k_{\cv} = \kappa^* - 1}} \left[\Pdist_{\H\mid v, A}(\val_{\H}(\cv; A)=\kappa^* - 1)-2\eta\right]^+ \leq \varepsilon |\Cs| \leq 1. 
    \end{equation}
    This gives an upper bound on the contribution from bad vertices in \cref{eqn:threshold-potential-sum}. Combining the contributions of both the good and bad vertices:
    \begin{equation}
        \Psi(t+1) - \Phi(t) \leq \frac{(n-\kappa^* + 1)2^{-(n-\kappa^* + 1)}}{M-t} \E_{A\in \Us_t}[L+1] = \frac{(n-\kappa^* + 1)2^{-(n-\kappa^* + 1)}}{M-t} (L+1),
    \end{equation}
    completing the proof.
\end{proof}

Finally, we complete the proof of our lower bound using \cref{lemma:success-bound,lemma:jensens-potential,lemma:potential-growth}:
\begin{theorem}\label{thm:yz-lower-bound}
    Consider the $\CI$ streaming problem (\cref{problem:YZ-streaming}), where $\bm{\pi}$ is a uniformly random ordering and $\Cs$ is the folded Reed--Solomon code family in \cref{prop:YZ-list-recovery}. There exist absolute constants $c_1, c_2>0$ such that for large enough $n$, any classical streaming algorithm using at most $2^{c_1 n^{c_2}}$ bits of memory must have success probability at most $2^{-\Omega(n)}$ over the randomness of the oracle and stream ordering.
\end{theorem}
\begin{proof}
    From \cref{lemma:success-bound}, $\Phi(0)=0$ and the success probability of any branching program is upper bounded by $2\eta + \Phi(M)$. We begin by bounding the potential difference at each time step $t$:
    \begin{equation}
        \Phi(M) - 0 = \sum_{t=0}^{M-1} \Phi(t+1) - \Phi(t) \leq \sum_{t=0}^{M-1} \Psi(t+1) - \Phi(t),
    \end{equation}
    where the last inequality is due to \cref{lemma:jensens-potential}. 
    
    Set $\gamma \gets 1/2, \eta \gets 2^{-(1-\rho) n/4}$, and $\kappa \gets \lceil\rho n + \log(1/\eta)/\gamma\rceil + 1 \leq n$.
    Since
    \begin{equation}
        S \leq 2^{c_1 n^{c_2}} \leq \ell/2 - O(n^2\log n) - (1-\rho) n/4
    \end{equation}
   for small enough constants $c_1,c_2 > 0$ and large enough $n$, we satisfy the conditions of \cref{lemma:potential-growth}. Applying the lemma and writing $r\triangleq n-\kappa+1 = \lfloor (1-\rho) n/2\rfloor$ gives,
    \begin{equation}
        2\eta + \Phi(M) \leq 2^{1-(1-\rho)n/4}+ (L+1)r2^{-r} \sum_{t=1}^M \frac{1}{t} \leq 2^{1-(1-\rho)n/4}+(L+1)r2^{-r} (1+\ln M) \leq 2^{-\Omega(n)}
    \end{equation}
    where we use $L \leq 2^{o(n)}$, $r=\Theta(n)$, and $\log M = \Theta(n \log n)$ due to \cref{prop:YZ-list-recovery}.
\end{proof}
\begin{remark}
    We can generalize the lower bound: for a general $(\rho, \ell, L)$-list-recoverable code where $\rho\in (0,1)$ and
    \begin{equation}
        \log |\Cs| + n \leq o(\ell), \qquad \log L + \log\log M = o(n),
    \end{equation}
    any classical algorithm solving $\CI$ with memory $o(\ell)$ has success probability at most $2^{-\Omega(n)}$.
\end{remark}
\section{Classical space lower bound for streaming Optimal Polynomial Intersection}\label{sec:OPI-lower-bound}
In this section we extend the techniques introduced in \cref{sec:YZ-lower-bound} and use them to prove our classical space lower bound for $\OPIshort$ in a random-order stream in \cref{thm:opi-lower-bound} (formal version of the lower bound in \cref{thm:opi-lower-bound-informal}).

\subsection{Extending the potential to constraint optimization}
We show how to extend our potential function from \cref{def:potential} to prove the more general setting where we allow codewords $\cv$ to have at least a certain number of their address mappings $\H_i(c_i)$ equal 1 instead of requiring all addresses to map to 1 as in $\CI$.

Given a codeword $\cv$, vertex $v\in V_t$ and observed address set $A$, the posterior probability that $\cv$ satisfies $\val_{\H}(\cv)\geq\tau$ again splits into two components: the probability of at least $j$ out of $|T_{\cv}\cap A|$ observed addresses mapping to 1, and the probability of \emph{exactly} $\tau - j$ out of $n-|T_{\cv}\cap A|$ unobserved addresses mapping to 1; we then sum over all possible $j$. Since oracle mappings on the unobserved addresses are independent of the branching program state, we get
\begin{equation}\label{eqn:general-conditional-prob}
    \Pdist_{\H\mid v,A}(\val_{\H}(\cv)\geq\tau)
    = \mathlarger{\sum}_{j=\tau - n + k_{\cv}}^{\min(k_{\cv}, \tau)} \underbrace{2^{-(n-k_{\cv})}\binom{n-k_{\cv}}{\tau - j}}_{\text{unobserved}}
    \underbrace{
    \vphantom{\binom{n-k_{\cv}}{j}}
    \Pdist_{\H\mid v,A}(\val_\H(\cv; A)\geq j)}_{\text{observed}}.
\end{equation}
The unobserved contribution still retains the $2^{-(n-k_{\cv})}$ factor, but now contains an additional binomial coefficient compared to \cref{eqn:YZ-conditional-prob}: the value $\binom{n-k_{\cv}}{\tau - j}$ counts the unobserved completion patterns with $\tau - j$ successes.

\begin{definition}[Generalized threshold crossing potential]\label{def:generalized-potential}
    Fix constants $\eta\in(0,1/2)$, $\tau\in[n]$, and function $\kappa: \{0,1,\ldots, n\}\rightarrow \Z_{\geq 1}$ satisfying
    \begin{equation}\label{eqn:opi-kappa-properties}
        \kappa(0)\geq 1,\quad
        \kappa(n)\leq\tau,\quad
        \kappa(k+1)-\kappa(k)\in\{0,1\}\ \text{ for } 0\leq k < n.
    \end{equation}
    We define the generalized potential function $\Phi_{\eta, \kappa}^{(\tau)}$ at time $t$ as
    \begin{equation}\label{eqn:generalized-potential}
        \Phi_{\eta, \kappa}^{(\tau)}(t)
        \triangleq \sum_{\substack{v\in V_t\\A\in \Us_t}} \Pr(E_{v,A})
        \sum_{\cv\in\Cs} \sum_{j\geq \kappa(k_{\cv})}
        2^{-(n-k_{\cv})}\binom{n-k_{\cv}}{\tau - j}
        \left[\Pdist_{\H\mid v,A}(\val_\H(\cv; A)\geq j)-2\eta\right]^+.
    \end{equation}
\end{definition}
\begin{remark}
    For all $v,A$, there are only a finite number of values $j$ in the summation which give a non-zero contribution: $j\leq \tau$ from the binomial coefficient and $j\leq k_{\cv}$ since $\val_\H(\cv; A)\leq k_{\cv}$. We do not include these upper bounds for simplicity.
\end{remark}
\begin{remark}
    To recover \cref{def:potential}, set $\tau \gets n$ and fix $\kappa$ to a constant function. The binomial coefficient $\binom{n-k_{\cv}}{\tau - j}$ ensures that $j\geq k_{\cv}$ for any non-zero contribution. The only term in the summation is $j=k_{\cv}$, giving the expression of the original potential.
\end{remark}

We again omit the parameters $\eta,\kappa$ when they are clear from context. For convenience, we define shorthand notation for binomial weighting that gives unobserved probability contribution in \cref{eqn:generalized-potential}:
\begin{equation}
    \beta_m(r) \triangleq \Pr(\mathrm{Binom}(m,1/2)=r) = 2^{-m}\binom{m}{r},
\end{equation}
where $\mathrm{Binom}(m,p)$ denotes the binomial random variable with $m$ trials and probability $p$. For our potential we can now write
\begin{equation}
    \Phi^{(\tau)}(t)= \sum_{\substack{v\in V_t\\A\in \Us_t}} \Pr(E_{v,A})
        \sum_{\cv\in\Cs} \sum_{j\geq\kappa(k_{\cv})}
       \beta_{n-k_{\cv}}(\tau - j)
        \left[\Pdist_{\H\mid v,A}(\val_\H(\cv; A)\geq j)-2\eta\right]^+.
\end{equation}
 
The generalized potential again begins at zero and bounds the final success probability.
\begin{lemma}[Generalized potential bounds $\OPIshort$ success probability]\label{lemma:opi-success-bound}
    For any $\eta\in(0,1/2)$ and function $\kappa:\{0,1,\ldots, n\}\rightarrow \Z_{\geq 1}$ satisfying \cref{eqn:opi-kappa-properties}, we have $\Phi^{(\tau)}(0)=0$ and
    \begin{equation}\label{eqn:opi-success-bound}
        \Pr(\Bs\ \mathrm{succeeds})\leq 2\eta+\Phi^{(\tau)}(M).
    \end{equation}
\end{lemma}
\begin{proof}
    At $t=0$, there are no observed addresses, so $k_{\cv}=0$ for every $\cv\in\Cs$. Since $\kappa(0)\geq 1$, $\Pdist_{\H\mid v}(\val_\H(\cv; A)\geq j) = 0$ for all $j\geq \kappa(0)$ and codewords $\cv$, implying $\Phi^{(\tau)}(0)=0$.
 
    At $t=M$, since $\kappa(n) \leq \tau$ for every $\cv\in\Cs$, the only non-zero binomial weight is $\beta_0(\tau-j)=1$ when $j=\tau$; the weights are zero for all other values $j\neq\tau$, so 
    \begin{equation}
        \Phi^{(\tau)}(M)=\sum_{v\in V_M}\Pr(E_v)\sum_{\cv\in\Cs}\left[\Pdist_{\H\mid v}(\val_{\H}(\cv)\geq\tau)-2\eta\right]^+.
    \end{equation}
    The remainder of the proof now follows \cref{lemma:success-bound} with the event $\{\val_{\H}(\cv)\geq\tau\}$ replacing $\{\val_{\H}(\cv)=n\}$:
    \begin{align}
        \Pr(\Bs\ \mathrm{succeeds})
        &=\sum_{v\in V_M}\Pr(E_v)\,\Pdist_{\H\mid v}(\val_{\H}(\cv_v)\geq\tau)\\
        & \leq 2\eta + \sum_{v\in V_M} \Pr(E_v) \max_{\cv\in \Cs} \left[ \Pdist_{\H\mid v}(\val_{\H}(\cv)\geq\tau) - 2\eta\right]^+ \\
        &\leq 2\eta + \sum_{v\in V_M} \Pr(E_v) \sum_{\cv\in\Cs} \left[\Pdist_{\H\mid v}(\val_{\H}(\cv)\geq\tau) - 2\eta\right]^+ \\
        &= 2\eta+\Phi^{(\tau)}(M).
    \end{align}
\end{proof}

\subsection{Bounding the potential for Optimal Polynomial Intersection}
Throughout this subsection, we consider the codeword search interpretation of $\OPIshort$ using \cref{def:code-relation}. Set $\Sigma\triangleq\F_q$ and $\Cs$ to the Reed--Solomon code $\RS_q(n,d)$. We now follow the same proof outline in \cref{sec:YZ-lower-bound} for proving \cref{thm:yz-lower-bound}: \cref{lemma:opi-jensens-potential,lemma:opi-pairwise-agreement,lemma:opi-potential-growth} below are the $\OPIshort$ analogues of \cref{lemma:jensens-potential,lemma:list-recovery-confidence,lemma:potential-growth}. 

We note two key differences in this new proof compared to the one presented in \cref{sec:YZ-lower-bound}:
\begin{enumerate}
    \item The threshold $\kappa$, besides being a function, now constrains the observed value $\val_{\H}(\cv;A)$ of a codeword rather than just its number of observed coordinates $k_{\cv}$. Previously the two were equivalent for any codeword contributing to the $\CI$ potential, but this is no longer the case for $\OPI$. 

    \item Unlike \cref{lemma:list-recovery-confidence}, \cref{lemma:opi-pairwise-agreement} does not rely on the density-restoring partition, and instead uses the fact that distinct polynomials agree on at most $d-1$ coordinates.
\end{enumerate}

Similarly to \cref{sec:YZ-lower-bound}, we define the edge-parameterized potential for all $1 \leq t \leq M$ by
\begin{equation}\label{eqn:opi-edge-potential}
    \Psi^{(\tau)}_{\eta,\kappa}(t)\triangleq\sum_{\substack{e\in\Gamma_t\\A\in\Us_t}}\Pr(E_{e,A})\sum_{\cv\in\Cs}\sum_{j\geq\kappa(k_{\cv})}\beta_{n-k_{\cv}}(\tau-j)\left[\Pdist_{\H\mid e,A}(\val_{\H}(\cv;A)\geq j)-2\eta\right]^+.
\end{equation}
 
\begin{lemma}[Convexity bound]\label{lemma:opi-jensens-potential}
    Fix constant $\eta\in(0,1/2)$, $\tau \in [n]$, and function $\kappa:\{0,1,\ldots, n\}\rightarrow \Z_{\geq 1}$. Then 
    \begin{equation}
        \Phi^{(\tau)}(t)\leq\Psi^{(\tau)}(t)
    \end{equation}
    for all $1 \leq t \leq M$.
\end{lemma}
\begin{proof}
    From an identical argument as \cref{lemma:jensens-potential}, we get
    \begin{equation}\label{eqn:opi-edge-decomposition}
        \Pr(E_{v,A})\left[\Pdist_{\H\mid v, A}(\val_{\H}(\cv; A)\geq j)-2\eta\right]^+ \leq \sum_{e\in \Gamma_{\text{in}}(v)}\Pr(E_{e,A}) \left[\Pdist_{\H\mid e, A}(\val_{\H}(\cv; A)\geq j)-2\eta\right]^+.
    \end{equation}
    For a fixed set of observed addresses $A\in\Us_t$, $k_{\cv}$ is specified for each codeword $\cv$. Therefore $A$ completely defines the binomial weights $\beta_{n-k_{\cv}}(\tau-j)$ and the range of $j$ which \cref{eqn:opi-edge-potential} sums over. We can again take the summation over $A\in\Us_t$, $v\in V_t$, and $\cv\in\Cs$ to both sides, as well as the summation over $j$ and binomial weight $\beta_{n-k_{\cv}}(\tau-j)$ to both sides of \cref{eqn:opi-edge-decomposition}, which gives the desired result.
\end{proof}

The proof of \cref{lemma:list-recovery-confidence} utilizes the density-restoring partition to show that the oracle is unlikely to map \emph{all} unfixed coordinates of a codeword to $1$. In the case of $\OPIshort$, we forgo the density-restoring partition and instead directly use the fact that distinct polynomials of degree less than $d$ have at most $d-1$ common points. We use this to show that only few codewords can be above the value threshold with non-negligible probability.

\begin{lemma}[Few high-value polynomials]\label{lemma:opi-pairwise-agreement}
For fixed $1\leq d\leq n$, let $\Cs\subseteq\Sigma^n$ be a code in which any two distinct codewords agree on at most $d-1$ coordinates. Fix $A\subseteq\Us$, and let $\Xm$ be a random variable on $\bits^A$ with $H_\infty(\Xm)\ge |A|-D$. For $\cv\in\Cs$, write $k_{\cv}\triangleq |T_{\cv}\cap A|$ and let $E_{\cv}$ be the event
\[
    |\Xm_{T_{\cv}\cap A}|\ge k_{\cv}/2+s.
\]
Suppose $s\ge 3\sqrt n$ and $D\le s^2/(8d)$, and set
$\eta\triangleq e^{-s^2/(4n)}$. Then fewer than $n/(d\eta)$ codewords
$\cv\in\Cs$ satisfy $\Pr[E_{\cv}]\ge 2\eta$. Consequently,
\begin{equation}
\label{eqn:opi-pairwise-agreement}
    \sum_{\cv\in\Cs}
    \left[\Pr[E_{\cv}]-2\eta\right]^+
    \le \frac{n}{d\eta}.
\end{equation}
\end{lemma}
\begin{proof}
    Fix integer $r\triangleq \lfloor n/d\rfloor$. Then $rd\leq n$ and $D\leq s^2/(8d)\le rs^2/(4n)$.

    We first show that any $r$ distinct codewords are unlikely to simultaneously have an observed Hamming weight surplus $s$ above their mean. Fix distinct
    $\cv_1,\ldots,\cv_r$, and for $u\in A$ let
    $w(u)\triangleq|\{i:u\in T_{\cv_i}\}|$. If all events
    $E_{\cv_1},\ldots,E_{\cv_r}$ occur, then
    \[
        \sum_{u\in A}w(u)(X_u-1/2)\ge rs.
    \]
    Since two distinct codewords agree on at most $d-1$ coordinates,
    $\sum_{u\in A}w(u)^2\le rn+r(r-1)(d-1)\le2rn$.
    Hence, for uniform independent bits on $A$, Hoeffding's inequality gives
    \[
        \Pr\left[\bigcap_{i=1}^r E_{\cv_i}\right]
        \le e^{-rs^2/n}.
    \]
    The min-entropy bound gives
    $\Pr(\Xm=\zv)\le2^{D-|A|}$ for every $\zv\in\bits^A$.
    Therefore, under $\Xm$,
    $\Pr[\bigcap_{i=1}^r E_{\cv_i}]
      \le2^D e^{-rs^2/n}$.

    Suppose now that at least $r/\eta$ codewords satisfy
    $\Pr(E_{\cv})\ge2\eta$. Choose $R$ of them with
    $r/\eta\le R\le2r/\eta$. The expected number of occurring events among these $R$ codewords is at least $2\eta R$. Since at most $R$ events can occur in total and $r\leq\eta R$, this expectation is impossible unless at least $r$ events occur with probability at least $\eta$.

    On the other hand, a union bound over the $r$-subsets gives
    \[
        \Pr[\text{at least $r$ events occur}]
        \le \binom{R}{r}2^D e^{-rs^2/n}
        \le \left(\frac{2e}{\eta}\right)^r
             2^D e^{-rs^2/n}
        < \eta.
    \]
    Indeed, $\eta=e^{-s^2/(4n)}$, $s\geq 3\sqrt n$, and
    $D\leq rs^2/(4n)$ imply that the last expression is at most
    $e^{-rs^2/(4n)}\leq\eta$. This is a contradiction.

    Thus fewer than $r/\eta\leq n/(d\eta)$ codewords satisfy
    $\Pr(E_{\cv})\geq 2\eta$. All other terms inside the $\max$ function vanish, while each remaining term is at most $1$. Hence
    \[
        \sum_{\cv\in\Cs}[\Pr(E_{\cv})-2\eta]^+ \leq \frac{n}{d\eta},
    \]
    as claimed.
\end{proof}

Now we move to bound the potential growth at each time step using \cref{lemma:opi-pairwise-agreement}. We remark that the threshold $\kappa$ in the following lemma is chosen analogously to \cref{lemma:potential-growth}: $\kappa(k)-1$ is the smallest integer that is at least the threshold $\frac{k}{2}+s$ in \cref{lemma:opi-pairwise-agreement}.

\begin{lemma}[Generalized potential growth]\label{lemma:opi-potential-growth}
    For fixed $1\leq d\leq n$, let $\Cs\subseteq\Sigma^n$ be a code where any two distinct codewords agree on at most $d-1$ coordinates. Fix $s\geq 3\sqrt n$, and let
    \begin{equation}\label{eqn:opi-kappa}
        \eta\triangleq e^{-s^2/(4n)} < 1/2
        \qquad\text{and}\qquad
        \kappa(k)\triangleq\left\lceil\frac k2+s\right\rceil+1\quad\text{for } 0\leq k\leq n.
    \end{equation}
    If the branching program $\Bs$ has width $w=2^S$, where
    $S\leq s^2/(8d)- \log|\Cs|$, then for every $0\leq t<M$,
    \begin{equation}
        \Psi^{(\tau)}(t+1)-\Phi^{(\tau)}(t)\leq\frac{n/(d\eta)+1}{M-t}\max_{0\leq k<n}\left(\frac{n-k}{2}\right)\beta_{n-k-1}\left(\tau-\kappa(k)\right).
    \end{equation}
\end{lemma}
\begin{proof}
    Following the argument and notation from \cref{lemma:potential-growth}, we partition the summation over edges based on the vertex each edge originates from:
    \begin{equation}\label{eqn:opi-edge-labeled-summation}
        \Psi^{(\tau)}(t+1) = \sum_{\substack{v\in V_t\\A\in \Us_t}} \Pr(E_{v,A}) \sum_{\substack{(i,x)\in A^C\\b\in\bits}} \sum_{\cv\in\Cs} \sum_{j\geq\kappa(k_{\cv}^+)}\frac{\beta_{n-k_{\cv}^+}(\tau-j)}{2(M-t)}\left[\Pdist_{\H\mid e,A^+}(\val_{\H}(\cv;A^+)\geq j)-2\eta\right]^+,
    \end{equation}
    where $e$ is the edge out of $v$ that reads $(i,x,b)$.

    \paragraph{Individual codeword analysis.} Fix $v\in V_t$ and $A\in\Us_t$ with $\Pr(E_{v,A})>0$, and consider a fixed codeword $\cv\in\Cs$. Abbreviate $k\triangleq k_{\cv}$, $k^+\triangleq k_{\cv}^+$, and $P(j)\triangleq\Pdist_{\H\mid v,A}(\val_{\H}(\cv;A)\geq j)$. Dividing by the common factor $\Pr(E_{v,A})$, the contribution of $\cv$ to $\Phi^{(\tau)}(t)$ is $\sum_{j\geq\kappa(k)}\beta_{n-k}(\tau-j)\left[P(j)-2\eta\right]^+$, and we compare it with the contributions of the pairs $(i,x)\in A^C$ to \cref{eqn:opi-edge-labeled-summation}. As in the proof of \cref{lemma:potential-growth}, conditioning on $E_{e,A^+}$ instead of $E_{v,A}$ leaves the distribution of $\H_A$ unchanged and fixes $\H_i(x)=b$, so $\Pdist_{\H\mid e,A^+}(\val_{\H}(\cv;A^+)\geq j)$ equals $P(j)$ if $x\neq c_i$ and $P(j-b)$ if $x=c_i$. We again have two cases for $(i,x)$:
    \begin{itemize}
        \item If $x\neq c_i$, then $(i,x)\in A^C\setminus T_{\cv}$ and $k^+=k$, so the observed address probability does not change. The contribution from each of the $M-t-(n-k)$ pairs $(i,x)$ is therefore
        \begin{equation}
           \frac{1}{M-t}\sum_{j\geq\kappa(k)}\beta_{n-k}(\tau-j)\left[P(j)-2\eta\right]^+.
        \end{equation}

        \item If $x=c_i$, then $(i,c_i)\in T_{\cv}\cap A^C$ and $k^+=k+1$. We sum the $b=0$ and $b=1$ contributions to get a contribution from each of the $n-k$ pairs of
        \begin{equation}\label{eqn:opi-codeword-term}
            \begin{aligned}
                \frac{1}{2(M-t)}\Bigg(&\sum_{j\geq\kappa(k+1)}\beta_{n-k-1}(\tau-j)\left[P(j)-2\eta\right]^+\\
                &+\sum_{j\geq\kappa(k+1)-1}\beta_{n-k-1}(\tau-j-1)\left[P(j)-2\eta\right]^+\Bigg).
            \end{aligned}
        \end{equation}
    \end{itemize}
    For $k<n$, Pascal's identity gives
    \begin{equation}
        \binom{n-k}{a} = \binom{n-k-1}{a} + \binom{n-k-1}{a-1} \implies \beta_{n-k}(a) = \frac{1}{2}\left(\beta_{n-k-1}(a)+\beta_{n-k-1}(a-1)\right)
    \end{equation}
    for every integer $a$, so
    \begin{equation}\label{eqn:pascal-sum}
        \begin{aligned}
            \sum_{j\geq\kappa(k)}2\beta_{n-k}(\tau-j)\left[P(j)-2\eta\right]^+ &= \sum_{j\geq\kappa(k)}\beta_{n-k-1}(\tau-j)\left[P(j)-2\eta\right]^+ \\&\qquad+ \sum_{j\geq\kappa(k)}\beta_{n-k-1}(\tau-j-1)\left[P(j)-2\eta\right]^+.
        \end{aligned}
    \end{equation}
    Since $\kappa(k+1)\in\{\kappa(k),\kappa(k)+1\}$ by \cref{eqn:opi-kappa}, we treat the two cases separately:
    \begin{itemize}
        \item If $\kappa(k+1)=\kappa(k)+1$, then $2(M-t)$ times \cref{eqn:opi-codeword-term} has one less term than the right-hand side of \cref{eqn:pascal-sum}, and therefore it is at most the left-hand side. In this case there is no positive contribution to the potential for this particular codeword.

        \item If $\kappa(k+1)=\kappa(k)$, then $2(M-t)$ times \cref{eqn:opi-codeword-term} has the extra term
        \begin{equation}
            \beta_{n-k-1}(\tau-\kappa(k))\left[P(\kappa(k)-1)-2\eta\right]^+
        \end{equation}
        compared to the right-hand side of \cref{eqn:pascal-sum}.
    \end{itemize}
    Summing over all $M-t$ pairs $(i,x)\in A^C$ and restoring the factor $\Pr(E_{v,A})$, the contribution of $\cv$ to $\Psi^{(\tau)}(t+1)-\Phi^{(\tau)}(t)$ is at most
    \begin{equation}\label{eqn:opi-threshold-term}
        \Pr(E_{v,A})\,\frac{n-k}{2(M-t)}\,\beta_{n-k-1}\big(\tau-\kappa(k)\big)\left[P(\kappa(k)-1)-2\eta\right]^+
    \end{equation}
    if $k<n$. Otherwise if $k=n$, the contribution is always $0$. The term in \cref{eqn:opi-threshold-term} is analogous to the contribution of codewords crossing the $k_{\cv}=\kappa^*$ threshold in \cref{lemma:potential-growth}.

    \paragraph{Summing over codewords that cross the value threshold.} Summing \cref{eqn:opi-threshold-term} over all $v\in V_t$, $A\in\Us_t$ and codewords $\cv\in\Cs$ with $k_{\cv}<n$,
    \begin{align}
        \Psi^{(\tau)}(t+1)-\Phi^{(\tau)}(t)
        \leq\E_{A\in\Us_t}\sum_{v\in V_t}\Pr(\bm{v_t}&=v\mid\bm{A_t}=A)\sum_{\substack{\cv\in\Cs\\k_{\cv}<n}}\frac{n-k_{\cv}}{2(M-t)}\beta_{n-k_{\cv}-1}\big(\tau-\kappa(k_{\cv})\big)\notag\\
        &\cdot\left[\Pdist_{\H\mid v,A}\big(\val_{\H}(\cv;A)\geq\kappa(k_{\cv})-1\big)-2\eta\right]^+.
    \end{align}

    Since
    \begin{equation}
        \frac{n-k_{\cv}}{2}\beta_{n-k_{\cv}-1}(\tau-\kappa(k_{\cv})) \leq \max_{0\leq k<n}\frac{n-k}{2}\,\beta_{n-k-1}\left(\tau-\kappa(k)\right)
    \end{equation}
    for any $\cv\in\Cs$, we apply it to each term in our summation over $\cv$ for a fixed $v,A$:
    \begin{equation}
        \max_{0\leq k<n}\left(\frac{n-k}{2}\right)\beta_{n-k-1}\left(\tau-\kappa(k)\right) \sum_{\substack{\cv\in\Cs\\k_{\cv}<n}} \left[\Pdist_{\H\mid v,A}\big(\val_{\H}(\cv;A)\geq\kappa(k_{\cv})-1\big)-2\eta\right]^+.
    \end{equation}
    
    Now we just need to show the remaining sum over $v$ and $\cv$ is upper bounded by $n/(d\eta)+1$. To do this, we can essentially follow the analogous proof in \cref{lemma:potential-growth}: we apply \cref{prop:min-entropy-deficiency} to partition into good and bad vertices with $\varepsilon\gets 1/|\Cs|$, \cref{lemma:opi-pairwise-agreement} gives a contribution of $n/(d\eta)$ on the good vertices and the trivial upper bound gives a contribution of $\varepsilon|C|=1$ on the bad vertices, completing the proof.
\end{proof}

\begin{theorem}[$\OPI$ lower bound]\label{thm:opi-lower-bound}
    Consider the $\OPIshort$ random-order streaming problem (\cref{problem:OPI-streaming}), where the degree $d=d(n)$ satisfies $d = \omega(\log n)$ and $d = o(n/\log n)$. Set target value $\tau=\lceil n/2+\sqrt{dn}/8\rceil$. There exists a constant $c>0$ such that, for all sufficiently large $n$, any classical streaming algorithm using at most $cn$ bits of memory has success probability at most $2^{-\Omega(d)}$ over the randomness of the oracle and stream ordering.
\end{theorem}
\begin{proof}
    Recall $q=n+1$, $M=n(n+1)$. Since two distinct polynomials of degree less than $d$ agree on at most $d-1$ evaluation points, two codewords of $\Cs=\RS_q(n,d)$ agree on at most $d-1$ coordinates. Let $\Bs$ be a branching program of length $M$ and width $2^S$ where $S\leq cn$, for some constant $c>0$. Define $\Delta \triangleq \tau - n/2$, so that $\Delta\geq\sqrt{dn}/8$. 
    
    Similarly to \cref{thm:yz-lower-bound}, we expand the start and end potential difference
    \begin{equation}
        \Phi^{(\tau)}(M) - \Phi^{(\tau)}(0) = \sum_{t=0}^{M-1} \Phi^{(\tau)}(t+1) - \Phi^{(\tau)}(t) \leq \sum_{t=0}^{M-1} \Psi^{(\tau)}(t+1) - \Phi^{(\tau)}(t),
    \end{equation}
    where the last inequality is due to \cref{lemma:opi-jensens-potential}. 
    
    We apply \cref{lemma:opi-potential-growth} with $s\gets \Delta/2$, which sets $\eta=e^{-\Delta^2/(16n)}$ and $\kappa(k)=\lceil(k+\Delta)/2\rceil+1$; from our assumption of $d(n)$, for all sufficiently large $n$, we have $s\geq 3\sqrt n$ and
    \begin{equation}
        s^2/(8d) -\log|\Cs| \geq \Omega(n) - d\log(n+1) \geq cn
    \end{equation}
    for a small enough constant $c$. Therefore all the conditions of \cref{lemma:opi-potential-growth} are satisfied, and we have
    \begin{equation}
        \Psi^{(\tau)}(t+1)-\Phi^{(\tau)}(t)\leq\frac{n/(d\eta)+1}{M-t}\max_{0\leq k<n}\left(\frac{n-k}{2}\right)\beta_{n-k-1}\left(\tau-\kappa(k)\right).
    \end{equation}

    Since $\kappa(n-1) \leq \frac{n + 3 + \Delta}{2} = \frac{n/2 + 3 + \tau}{2} < \tau$ for large enough $n$, $\beta_0(\tau-\kappa(n-1)) = 0$. For all $k\leq n-2$, applying Hoeffding's inequality (\cref{prop:hoeffding}) with $\delta \gets \tau - \kappa(k) - (n-k-1)/2 \geq \Delta/4$ and $m \gets n-k-1$ gives, for large enough $n$,
    \begin{equation}
        \beta_{n-k-1}(\tau-\kappa(k))\leq \Pr(\mathrm{Binom}(n-k-1, 1/2)\geq \tau-\kappa(k))\leq e^{-\Delta^2/(8n)}.
    \end{equation}
    This implies
    \begin{equation}\label{eqn:potential-growth-bound}
         \Psi^{(\tau)}(t+1)-\Phi^{(\tau)}(t)\leq \frac{n/(d\eta)+1}{M-t}\max_{0\leq k<n}\frac{n-k}{2}\,\beta_{n-k-1}\left(\tau-\kappa(k)\right) \leq \frac{n/(d\eta)+1}{M-t}\cdot\frac{n}{2} e^{-\Delta^2/(8n)}.
    \end{equation}
    Since $\kappa(n) \leq \frac{n + 4 + \Delta}{2} \leq \tau$ for large enough $n$, we can also use \cref{lemma:opi-success-bound}. This gives $\Phi^{(\tau)}(0)=0$ and $\Pr(\Bs\ \mathrm{succeeds})\leq 2\eta+\Phi^{(\tau)}(M)$, so
    \begin{equation}\label{eqn:penultimate-success}
        \Pr(\Bs\ \mathrm{succeeds}) \leq 2\eta + \sum_{t=0}^{M-1} \Psi^{(\tau)}(t+1) - \Phi^{(\tau)}(t).
    \end{equation}
    Plugging \cref{eqn:potential-growth-bound} into \cref{eqn:penultimate-success} and using $n/(d\eta)+1\leq 2n/(d\eta)=(2n/d)e^{\Delta^2/(16n)}$,
    \begin{equation}
        \Pr(\Bs\ \mathrm{succeeds}) \leq 2\eta + \sum_{t=0}^{M-1} \Psi^{(\tau)}(t+1) - \Phi^{(\tau)}(t) \leq 2\eta + \sum_{t=0}^{M-1}\frac{2ne^{\Delta^2/(16n)}}{d(M-t)} \cdot \frac{n}{2} e^{-\Delta^2/(8n)}
    \end{equation}
    Since $\eta = e^{-\Delta^2/(16n)}$ and $\sum_{t=1}^M 1/t\leq 1+\ln M$, we get
    \begin{equation}
        \Pr(\Bs\ \mathrm{succeeds}) \leq 2\eta + \frac{\eta n^2}{d} \sum_{t=1}^M \frac{1}{t} \leq \eta\left(2+\frac{n^2(1+\ln M)}{d}\right).
    \end{equation}
    Since $\Delta\geq\sqrt{dn}/8$, we have $\eta\leq e^{-d/1024}$. Additionally, $M=n(n+1)$ and $d\geq1$ implies $2+n^2(1+\ln M)/d\leq n^3$ for all $n\geq 10$. Therefore
    \begin{equation}\label{eqn:opi-final-absorption}
        \Pr(\Bs\ \mathrm{succeeds}) \leq n^3 e^{-d/1024} = \exp\left(-\frac{d}{1024}+3\ln n\right) \leq e^{-d/2048} = 2^{-\Omega(d)},
    \end{equation}
    where the last inequality holds because $d=\omega(\log n)$ implies $3\ln n\leq d/2048$ for all sufficiently large $n$. This completes the proof.
\end{proof}

\DeclareUrlCommand{\Doi}{\urlstyle{sf}}
\renewcommand{\path}[1]{\small\Doi{#1}}
\renewcommand{\url}[1]{\href{#1}{\small\Doi{#1}}}
\bibliographystyle{alphaurl}
\bibliography{refs}

@Article{CFKMP21,
  author  = {Chattopadhyay, Arkadev and Filmus, Yuval and Koroth, Sajin and Meir, Or and Pitassi, Toniann},
  journal = {SIAM Journal on Computing},
  title   = {{{Query-to-Communication Lifting Using Low-Discrepancy Gadgets}}},
  year    = {2021},
  number  = {1},
  pages   = {171--210},
  volume  = {50},
  doi     = {10.1137/19M1310153},
}

@Article{GPW20,
  author   = {G\"{o}\"{o}s, Mika and Pitassi, Toniann and Watson, Thomas},
  journal  = {SIAM Journal on Computing},
  title    = {Query-to-communication lifting for {$\sf{BPP}$}},
  year     = {2020},
  number   = {4},
  pages    = {441--461},
  volume   = {49},
  doi      = {10.1137/17M115339X},
}

@InProceedings{KR11,
  author    = {Klartag, Bo'az and Regev, Oded},
  booktitle = {Proceedings of the 43rd Symposium on Theory of Computing (STOC)},
  title     = {Quantum one-way communication can be exponentially stronger than classical communication},
  year      = {2011},
  pages     = {31--40},
  publisher = {ACM, New York},
  doi       = {10.1145/1993636.1993642},
}

@InProceedings{Liu23,
  author    = {Qipeng Liu},
  booktitle = {Proceedings of the 42nd Conference on the Theory and Applications of Cryptographic Techniques (EUROCRYPT)},
  title     = {{{Non-uniformity and Quantum Advice in the Quantum Random Oracle Model}}},
  year      = {2023},
  pages     = {117--143},
  publisher = {Springer},
  series    = {Lecture Notes in Computer Science},
  volume    = {14004},
  doi       = {10.1007/978-3-031-30545-0_5},
}

@article{raz2018fast,
  title={Fast learning requires good memory: A time-space lower bound for parity learning},
  author={Raz, Ran},
  journal={Journal of the ACM (JACM)},
  volume={66},
  number={1},
  pages={1--18},
  year={2018},
  publisher={ACM New York, NY, USA}
}

@Article{Yamakawa2022,
  author    = {Yamakawa, Takashi and Zhandry, Mark},
  journal   = {Journal of the ACM},
  title     = {Verifiable quantum advantage without structure},
  year      = {2024},
  number    = {3},
  pages     = {1--50},
  volume    = {71},
  doi       = {10.1145/3658665},
  publisher = {ACM},
}

@inproceedings {GargRazTal18,
    AUTHOR = {Garg, Sumegha and Raz, Ran and Tal, Avishay},
     TITLE = {Extractor-based time-space lower bounds for learning},
 BOOKTITLE = {S{TOC}'18---{P}roceedings of the 50th {A}nnual {ACM} {SIGACT}
              {S}ymposium on {T}heory of {C}omputing},
     PAGES = {990--1002},
 PUBLISHER = {ACM, New York},
      YEAR = {2018},
      ISBN = {978-1-4503-5559-9},
   MRCLASS = {68Q32},
  MRNUMBER = {3826311},
}

@inproceedings{goos2025quantum,
  title={Quantum {C}ommunication {A}dvantage in {TFNP}},
  author={G{\"o}{\"o}s, Mika and Gur, Tom and Jain, Siddhartha and Li, Jiawei},
  booktitle={Proceedings of the 57th Annual ACM Symposium on Theory of Computing},
  pages={1465--1475},
  year={2025}
}

@article{jordan2025dqi,
  title={Optimization by decoded quantum interferometry},
  author={Jordan, Stephen P and Shutty, Noah and Wootters, Mary and Zalcman, Adam and Schmidhuber, Alexander and King, Robbie and Isakov, Sergei V and Khattar, Tanuj and Babbush, Ryan},
  journal={Nature},
  volume={646},
  number={8086},
  pages={831--836},
  year={2025},
  publisher={Nature Publishing Group UK London},
  doi = {10.1038/s41586-025-09527-5}
}

@misc{khattar2025optimized,
      title={{{Verifiable Quantum Advantage via Optimized DQI Circuits}}}, 
      author={Tanuj Khattar and Noah Shutty and Craig Gidney and Adam Zalcman and Noureldin Yosri and Dmitri Maslov and Ryan Babbush and Stephen P. Jordan},
      year={2025},
      eprint={2510.10967},
      archivePrefix={arXiv},
      primaryClass={quant-ph}
}

@misc{zhao2026exponential,
  title = {Exponential quantum advantage in processing massive classical data},
  author = {Haimeng Zhao and Alexander Zlokapa and Hartmut Neven and Ryan Babbush and John Preskill and Jarrod R. McClean and Hsin-Yuan Huang},
  year = {2026},
  eprint = {2604.07639},
  archiveprefix = {arXiv},
  primaryclass = {quant-ph},
}

@inproceedings{kallaugher2024exponential,
  title={Exponential quantum space advantage for approximating maximum directed cut in the streaming model},
  author={Kallaugher, John and Parekh, Ojas and Voronova, Nadezhda},
  booktitle={Proceedings of the 56th Annual ACM Symposium on Theory of Computing},
  pages={1805--1815},
  year={2024}
}

@misc{wang2026exponential,
      title={{{Exponential Quantum Space Advantage for Approximating Max-$k$SAT in the Streaming Setting}}}, 
      author={Haoyu Wang and Guangxu Yang},
      year={2026},
      eprint={2606.05366},
      archivePrefix={arXiv},
      primaryClass={cs.DS}
}

@InProceedings{moshkovitz17a,
  title = 	 {{{Mixing Implies Lower Bounds for Space Bounded Learning}}},
  author = 	 {Moshkovitz, Dana and Moshkovitz, Michal},
  booktitle = 	 {Proceedings of the 2017 Conference on Learning Theory},
  pages = 	 {1516--1566},
  year = 	 {2017},
  editor = 	 {Kale, Satyen and Shamir, Ohad},
  volume = 	 {65},
  series = 	 {Proceedings of Machine Learning Research},
  month = 	 {07--10 Jul},
  publisher =    {PMLR}
}

@article{dodis2008fuzzy,
  title={Fuzzy extractors: How to generate strong keys from biometrics and other noisy data},
  author={Dodis, Yevgeniy and Ostrovsky, Rafail and Reyzin, Leonid and Smith, Adam},
  journal={SIAM Journal on Computing},
  volume={38},
  number={1},
  pages={97--139},
  year={2008},
  publisher={SIAM}
}

@article{guha2009stream,
  title={{{Stream order and order statistics: Quantile estimation in random-order streams}}},
  author={Guha, Sudipto and McGregor, Andrew},
  journal={SIAM Journal on Computing},
  volume={38},
  number={5},
  pages={2044--2059},
  year={2009},
  publisher={SIAM}
}

@article{CCM2016robust,
 author = {Chakrabarti, Amit and Cormode, Graham and McGregor, Andrew},
 title = {{{Robust Lower Bounds for Communication and Stream Computation}}},
 year = {2016},
 pages = {1--35},
 doi = {10.4086/toc.2016.v012a010},
 publisher = {Theory of Computing},
 journal = {Theory of Computing},
 volume = {12},
 number = {10}
}

@inproceedings{saxena2023streaming,
  title={Streaming complexity of {CSP}s with randomly ordered constraints},
  author={Saxena, Raghuvansh R and Singer, Noah and Sudan, Madhu and Velusamy, Santhoshini},
  booktitle={Proceedings of the 2023 Annual ACM-SIAM Symposium on Discrete Algorithms (SODA)},
  pages={4083--4103},
  year={2023},
  organization={SIAM}
}

@article{Shor_1997,
   title={{{Polynomial-Time Algorithms for Prime Factorization and Discrete Logarithms on a Quantum Computer}}},
   volume={26},
   ISSN={1095-7111},
   DOI={10.1137/s0097539795293172},
   number={5},
   journal={SIAM Journal on Computing},
   publisher={Society for Industrial & Applied Mathematics (SIAM)},
   author={Shor, Peter W.},
   year={1997},
   month={10}, pages={1484–1509} }

@article{HHL_2009,
   title={{{Quantum Algorithm for Linear Systems of Equations}}},
   volume={103},
   ISSN={1079-7114},
   DOI={10.1103/physrevlett.103.150502},
   number={15},
   journal={Physical Review Letters},
   publisher={American Physical Society (APS)},
   author={Harrow, Aram W. and Hassidim, Avinatan and Lloyd, Seth},
   year={2009},
   month={10}}

@inproceedings{Raz99exponential,
author = {Raz, Ran},
title = {Exponential separation of quantum and classical communication complexity},
year = {1999},
isbn = {1581130678},
publisher = {Association for Computing Machinery},
address = {New York, NY, USA},
doi = {10.1145/301250.301343},
booktitle = {Proceedings of the Thirty-First Annual ACM Symposium on Theory of Computing},
pages = {358–367},
numpages = {10},
location = {Atlanta, Georgia, USA},
series = {STOC '99}
}

@article{gavinsky2009exponential,
author = {Gavinsky, Dmitry and Kempe, Julia and Kerenidis, Iordanis and Raz, Ran and de Wolf, Ronald},
title = {{{Exponential Separation for One-Way Quantum Communication Complexity, with Applications to Cryptography}}},
journal = {SIAM Journal on Computing},
volume = {38},
number = {5},
pages = {1695-1708},
year = {2009},
doi = {10.1137/070706550}
}

@misc{kretschmer2026unconditional,
      title={Demonstrating an unconditional separation between quantum and classical information resources}, 
      author={William Kretschmer and Sabee Grewal and Matthew DeCross and Justin A. Gerber and Kevin Gilmore and Dan Gresh and Nicholas Hunter-Jones and Karl Mayer and Brian Neyenhuis and David Hayes and Scott Aaronson},
      year={2026},
      eprint={2509.07255},
      archivePrefix={arXiv},
      primaryClass={quant-ph}
}

@inproceedings{CJP2008selection,
author = {Chakrabarti, Amit and Jayram, T. S. and P\u{a}tra\c{s}cu, Mihai},
title = {Tight lower bounds for selection in randomly ordered streams},
year = {2008},
publisher = {Society for Industrial and Applied Mathematics},
address = {USA},
booktitle = {Proceedings of the Nineteenth Annual ACM-SIAM Symposium on Discrete Algorithms},
pages = {720–729},
numpages = {10},
location = {San Francisco, California},
series = {SODA '08}
}

@misc{bostanci2026advice,
      title={Separating Quantum and Classical Advice with Good Codes}, 
      author={John Bostanci and Andrew Huang and Vinod Vaikuntanathan},
      year={2026},
      eprint={2602.09385},
      archivePrefix={arXiv},
      primaryClass={quant-ph}
}

@incollection {Raz17,
    AUTHOR = {Raz, Ran},
     TITLE = {A time-space lower bound for a large class of learning
              problems},
 BOOKTITLE = {58th {A}nnual {IEEE} {S}ymposium on {F}oundations of
              {C}omputer {S}cience---{FOCS} 2017},
     PAGES = {732--742},
 PUBLISHER = {IEEE Computer Soc., Los Alamitos, CA},
      YEAR = {2017},
      ISBN = {978-1-5386-3464-6},
   MRCLASS = {68Q17 (68Q87 68T05)},
  MRNUMBER = {3734276},
       DOI = {10.1109/FOCS.2017.73},
       URL = {https://doi.org/10.1109/FOCS.2017.73},
}

@incollection {Moshkovitz18,
    AUTHOR = {Moshkovitz, Dana and Moshkovitz, Michal},
     TITLE = {Entropy samplers and strong generic lower bounds for space
              bounded learning},
 BOOKTITLE = {9th {I}nnovations in {T}heoretical {C}omputer {S}cience},
    SERIES = {LIPIcs. Leibniz Int. Proc. Inform.},
    VOLUME = {94},
     PAGES = {Art. No. 28, 20},
 PUBLISHER = {Schloss Dagstuhl. Leibniz-Zent. Inform., Wadern},
      YEAR = {2018},
      ISBN = {978-3-95977-060-6},
   MRCLASS = {68Q32},
  MRNUMBER = {3761764},
}

@article{watrous1999space,
  title={Space-bounded quantum complexity},
  author={Watrous, John},
  journal={Journal of Computer and System Sciences},
  volume={59},
  number={2},
  pages={281--326},
  year={1999},
  publisher={Elsevier}
}

@inproceedings{le2006exponential,
  title={Exponential separation of quantum and classical online space complexity},
  author={Le Gall, Fran{\c{c}}ois},
  booktitle={Proceedings of the Eighteenth Annual ACM Symposium on Parallelism in Algorithms and Architectures},
  pages={67--73},
  year={2006}
}

@inproceedings{ta2013inverting,
author = {Ta-Shma, Amnon},
title = {Inverting well conditioned matrices in quantum logspace},
year = {2013},
isbn = {9781450320290},
publisher = {Association for Computing Machinery},
address = {New York, NY, USA},
url = {https://doi.org/10.1145/2488608.2488720},
doi = {10.1145/2488608.2488720},
booktitle = {Proceedings of the Forty-Fifth Annual ACM Symposium on Theory of Computing},
pages = {881–890},
numpages = {10},
location = {Palo Alto, California, USA},
series = {STOC '13}
}

@inproceedings{kallaugher2022quantum,
  title={A quantum advantage for a natural streaming problem},
  author={Kallaugher, John},
  booktitle={2021 IEEE 62nd Annual Symposium on Foundations of Computer Science (FOCS)},
  pages={897--908},
  year={2022},
  organization={IEEE}
}

@inproceedings{KKS14,
author = {Michael Kapralov and Sanjeev Khanna and Madhu Sudan},
title = {Approximating matching size from random streams},
booktitle = {Proceedings of the Twenty-Fifth Annual ACM-SIAM Symposium on Discrete Algorithms (SODA)},
pages = {734-751},
year={2014},
publisher = {Society for Industrial and Applied Mathematics},
doi = {10.1137/1.9781611973402.55}
}

@inproceedings{KKS15,
author = {Kapralov, Michael and Khanna, Sanjeev and Sudan, Madhu},
title = {{{Streaming lower bounds for approximating MAX-CUT}}},
year = {2015},
publisher = {Society for Industrial and Applied Mathematics},
booktitle = {Proceedings of the Twenty-Sixth Annual ACM-SIAM Symposium on Discrete Algorithms},
pages = {1263–1282},
numpages = {20},
location = {San Diego, California},
series = {SODA '15}
}

@InProceedings{assadi2021random,
  author =	{Assadi, Sepehr and Behnezhad, Soheil},
  title =	{{Beating Two-Thirds For Random-Order Streaming Matching}},
  booktitle =	{48th International Colloquium on Automata, Languages, and Programming (ICALP 2021)},
  pages =	{19:1--19:13},
  series =	{Leibniz International Proceedings in Informatics (LIPIcs)},
  ISBN =	{978-3-95977-195-5},
  ISSN =	{1868-8969},
  year =	{2021},
  volume =	{198},
  editor =	{Bansal, Nikhil and Merelli, Emanuela and Worrell, James},
  publisher =	{Schloss Dagstuhl -- Leibniz-Zentrum f{\"u}r Informatik},
  address =	{Dagstuhl, Germany},
  doi =		{10.4230/LIPIcs.ICALP.2021.19},
}

@article{munro1980selection,
  title={Selection and sorting with limited storage},
  author={Munro, J Ian and Paterson, Mike S},
  journal={Theoretical Computer Science},
  volume={12},
  number={3},
  pages={315--323},
  year={1980},
  publisher={Elsevier}
}

@InProceedings{hashemi2024weighted,
  author =	{Hashemi, Diba and Wrzos-Kaminska, Weronika},
  title =	{{Weighted Matching in the Random-Order Streaming and Robust Communication Models}},
  booktitle =	{Approximation, Randomization, and Combinatorial Optimization. Algorithms and Techniques (APPROX/RANDOM 2024)},
  pages =	{16:1--16:26},
  series =	{Leibniz International Proceedings in Informatics (LIPIcs)},
  ISBN =	{978-3-95977-348-5},
  ISSN =	{1868-8969},
  year =	{2024},
  volume =	{317},
  editor =	{Kumar, Amit and Ron-Zewi, Noga},
  publisher =	{Schloss Dagstuhl -- Leibniz-Zentrum f{\"u}r Informatik},
  address =	{Dagstuhl, Germany},
  doi =		{10.4230/LIPIcs.APPROX/RANDOM.2024.16}
}

@INPROCEEDINGS{chiplunkar2022factorial,
  author={Chiplunkar, Ashish and Kallaugher, John and Kapralov, Michael and Price, Eric},
  booktitle={2022 IEEE 63rd Annual Symposium on Foundations of Computer Science (FOCS)}, 
  title={{{Factorial Lower Bounds for (Almost) Random Order Streams}}}, 
  year={2022},
  pages={486-497},
  doi={10.1109/FOCS54457.2022.00053}}

@InProceedings{alexandru2026unit,
  author =	{Alexandru, Cezar-Mihail and Diddapur, Adithya and Halld\'{o}rsson, Magn\'{u}s M. and Konrad, Christian and Naidu, Kheeran K.},
  title =	{{Unit Interval Selection in Random Order Streams}},
  booktitle =	{43rd International Symposium on Theoretical Aspects of Computer Science (STACS 2026)},
  pages =	{4:1--4:20},
  series =	{Leibniz International Proceedings in Informatics (LIPIcs)},
  ISBN =	{978-3-95977-412-3},
  ISSN =	{1868-8969},
  year =	{2026},
  volume =	{364},
  publisher =	{Schloss Dagstuhl -- Leibniz-Zentrum f{\"u}r Informatik},
  address =	{Dagstuhl, Germany},
  doi =		{10.4230/LIPIcs.STACS.2026.4},
}

@article{Yossef04,
author = {Bar-Yossef, Ziv and Jayram, T. S. and Kerenidis, Iordanis},
title = {{{Exponential Separation of Quantum and Classical One-Way Communication Complexity}}},
journal = {SIAM Journal on Computing},
volume = {38},
number = {1},
pages = {366-384},
year = {2008},
doi = {10.1137/060651835}}

@article{montanaro2016frequency,
  author    = {Ashley Montanaro},
  title     = {The quantum complexity of approximating the frequency moments},
  journal   = {Quantum Information \& Computation},
  volume    = {16},
  number    = {13\&14},
  pages     = {1169--1190},
  year      = {2016}
}

@misc{feng2026entropy,
      title={Exponential quantum space advantage for {Shannon} entropy estimation in data streams}, 
      author={Weijun Feng and Yongzhen Xu and Lvzhou Li and Gongde Guo and Song Lin},
      year={2026},
      eprint={2604.18014},
      archivePrefix={arXiv},
      primaryClass={quant-ph}
}

@misc{WuYang2026,
      title={Exponential quantum advantages for decoded quantum interferometry in the streaming setting}, 
      author={Kewen Wu and Guangxu Yang},
      year={2026},
      eprint={2610.01902},
      archivePrefix={arXiv},
      primaryClass={quant-ph}
}

@misc{voronova2026random,
      title={{{Random Order in Quantum Streaming: Replenishment and Robust Lower Bounds}}}, 
      author={Nadezhda Voronova},
      year={2026},
      eprint={2610.06727},
      archivePrefix={arXiv},
      primaryClass={quant-ph}
}

\newpage
\appendix
\crefalias{section}{appendix}
\section{Quantum streaming algorithms for Code Intersection and Optimal Polynomial Intersection}
\label{sec:algorithms}
Both algorithms use the stream only to apply commuting input-controlled gates to a small
state. After the pass, they perform post-processing which is independent of the input. This gives us our single-pass, low space quantum algorithm. We specify the details below.

\begin{theorem}[Yamakawa--Zhandry streaming algorithm]\label{thm:YZ-algorithm}
    Let $\Cs_n$ be the code family specified in \cref{prop:YZ-list-recovery}. There exists a quantum streaming algorithm solving \cref{problem:YZ-streaming} using $\poly(n)$ memory with success probability $1-o(1)$ over the randomness of the oracle, for all stream orderings.
\end{theorem}
\begin{proof}
Let $Q=|\Sigma|$ and $F_i=\{a\in\Sigma:H_i(a)=1\},
  \ket{\tau_i}
  =\frac1{\sqrt{|F_i|}}\sum_{a\in F_i}\ket a$.
The Yamakawa--Zhandry algorithm succeeds with probability
$1-\negl(n)$ from
\[
  \ket{\tau}=\bigotimes_{i=1}^n\ket{\tau_i},
\]
using no further oracle access~\cite{Yamakawa2022,Liu23}. It therefore suffices to prepare
$\ket{\tau}$ in one pass.

Fix $k=C\log n$ for a sufficiently large constant $C$. For each
$i\in[n]$, prepare $k$ copies of
\[
  \frac1{\sqrt Q}\sum_{a\in\Sigma}\ket a\ket 0.
\]
When the stream item $(i,a,b)$ arrives, apply $\ket a\ket z\longmapsto\ket a\ket{z\oplus b}$
to every copy for coordinate $i$, controlled on its first register
being equal to $a$. Since every pair $(i,a)$ appears exactly once,
after the pass every copy for coordinate $i$ is
\[
  \frac1{\sqrt Q}\sum_{a\in\Sigma}\ket a\ket{H_i(a)}.
\]

Measure the last qubit of every copy. Conditional on outcome $1$,
the first register is exactly $\ket{\tau_i}$. Let
$p_i=|F_i|/Q$. Since each value \(H_i(a)\) is an independent fair bit, the number of \(a\in\Sigma\) with \(H_i(a)=1\) is concentrated around \(Q/2\).
Hence, except with probability $ne^{-\Omega(Q)}$, every $p_i\ge1/3$.
Conditioned on this event, the probability that all $k$ copies fail
for some coordinate is at most
\begin{equation}
    \sum_{i=1}^n (1-p_i)^k\leq n(2/3)^k=o(1).
\end{equation}
Thus, with probability $1-o(1)$ we obtain one copy of $\ket{\tau_i}$
for every $i$, and hence $\ket{\tau}$ exactly. We then run the
Yamakawa--Zhandry postprocessing.

The algorithm stores $O(nk\log Q)=O(n\log n\log Q)=\poly(n)$
qubits. It stores no oracle entry and works for every stream order.
\end{proof}

\begin{theorem}[DQI streaming algorithm]\label{thm:DQI-algorithm}
    Suppose $d=\omega(1)$, $d=o(n)$, and set
\[
    \tau\gets\left\lceil \frac n2+\frac18\sqrt{dn}\right\rceil .
\]
There exists a one-pass quantum streaming algorithm for
\cref{problem:OPI-streaming} using $O(d\log n)$ memory with success probability $1-o(1)$ over the randomness of the oracle, for all stream orderings.
\end{theorem}
\begin{proof}
   We will use the space-optimized version of the DQI algorithm \cite{jordan2025dqi} due to Khattar et al.~\cite{khattar2025optimized}. The idea is to store at most $d/2$ error locations, write all input-dependent phases during the pass, and decode only after the stream ends.

Set $\ell=\left\lfloor\frac{d-1}{2}\right\rfloor$, $\eta=\frac{d}{16n}$, $\rho=\Pr[\mathrm{Binom}(n,\eta)\le\ell]$, and, for $|I|\le\ell$, let
\begin{equation}
     c_I\triangleq\frac{\eta^{|I|/2}(1-\eta)^{(n-|I|)/2}}{\sqrt\rho}.
\end{equation}
Before the stream, prepare the state
\[
  \sum_{|I|\le\ell} c_I\ket I
  \bigotimes_{i\in I}\frac1{\sqrt q}\sum_{a\in\F_q}\ket a,
\]
storing $I$ as a sorted list. This state can be prepared within the same
$O(d\log n)$ workspace: scan $i=1,\ldots,n$, keep the current support size,
and use the conditional binomial probabilities to decide coherently whether to
append $i$ to the list. At all times we store only the sparse list and
$O(\log n)$ additional work.

Write $h_i(a)=2\H_i(a)-1$. When $(i,a,b)$ arrives, multiply by $2b-1$ every
branch whose sparse list contains $i$ and whose field register there equals
$a$. The gates commute. After one pass, each active location $i$ contains
$q^{-1/2}\sum_a h_i(a)\ket a$. No later step accesses $\H$.

Apply the Fourier transform over $\F_q$ to each active field register. Write
\begin{equation}
     \widehat h_i(y)=\frac1q\sum_{a\in\F_q}h_i(a)(-1)^{\Tr(ay)},
  \qquad \mu_i=\widehat h_i(0),
\end{equation}
where $\Tr: \F_q\rightarrow\F_2$ denotes the field trace.
We next compute one flag that is $1$ exactly when every active Fourier
register is nonzero, and measure only this flag. We do not measure $I$ or the
Fourier values. On failure the algorithm aborts. On success, the
\emph{unnormalized} successful branch is
\begin{equation}
    \sum_{|I|\le\ell}c_I \sum_{y_i\ne0\ (i\in I)}
  \left(\prod_{i\in I}\widehat h_i(y_i)\right)\ket{I,(y_i)_{i\in I}}.
\end{equation}
Let $e\in\F_q^n$ be the corresponding sparse vector, so
$\operatorname{wt}(e)\le\ell$.

Let $B_{i,j}=\alpha_i^j$ for $0\le j<d$. Compute the syndrome $B^Te$ into $d$ field registers. It determines $e$: if $B^Te=B^Te'$, then $e-e'$ has weight at most $2\ell<d$, while every set of fewer than $d$ columns of the Vandermonde matrix $B^T$ is independent. We can therefore erase $e$ coherently from its syndrome. For a self-contained space bound, enumerate all canonical sparse errors $e'$ of weight at most $\ell$. For each candidate, compute $B^Te'$ into scratch space, compare it with the stored syndrome, and, on the unique match, XOR $e'$ into the sparse error register. Then uncompute the scratch space. The candidate, its syndrome, and the loop counters use $O(d\log n)$ workspace. This may take exponential time, but our model only charges space.

Finally, inverse-Fourier-transform the syndrome and measure $f\in\Ps_d$. For a fixed oracle $\H$, Fourier inversion gives
\begin{equation}
  \Pr[\text{projection succeeds and output is }f\mid\H]
  =q^{-d}|G_\H(f)|^2,
  \label{eq:dqi-output-probability}
\end{equation}
where
\begin{equation}
  G_\H(f)=\sum_{|I|\le\ell}c_I
  \prod_{i\in I}\bigl(h_i(f(\alpha_i))-\mu_i\bigr).
  \label{eq:dqi-stream-amplitude}
\end{equation}
Since we leave the successful branch unnormalized,
$q^{-d}|G_\H(f)|^2$ is the probability that the projection succeeds and the
algorithm outputs $f$. If the projection fails, the algorithm aborts and
contributes zero. Therefore
\begin{equation}
    \Pr[\text{algorithm succeeds}]
  =\E_\H\left[q^{-d}\sum_{\OPIval(f)\ge\tau}|G_\H(f)|^2\right].
\end{equation}

For the success analysis, we compare with the ideal DQI amplitudes from
\cite{jordan2025dqi},
\[
  F_\H(f)=\prod_{i=1}^n
  \left(\sqrt{1-\eta}+\sqrt\eta\,h_i(f(\alpha_i))\right).
\]
Following \cite{jordan2025dqi}, average jointly over the random oracle $\H$
and a uniform polynomial $f$. Since $\E_{\H,f}F_\H(f)^2=1$, the weights
$F_\H(f)^2$ define a probability distribution on this joint experiment. Under
this distribution, the score is concentrated around
$n/2+(1/4+o(1))\sqrt{dn}$, and the probability of score below $\tau$ is
$e^{-\Omega(d)}$.

Our streaming state differs from this ideal state in two ways. First, we
truncate to supports of size at most $\ell$. The DQI analysis shows that this
costs $e^{-\Omega(d)}$ in squared $L_2$ norm. Second, the projection removes the
zero Fourier frequency. Let
$F_{\ell,\H}(f)=\sum_{|I|\le\ell}c_I\prod_{i\in I}h_i(f(\alpha_i))$.
For $X_i=h_i(f(\alpha_i))$ and $Y_i=X_i-\mu_i$,
$\E X_i^2=1$ and $\E Y_i^2=\E[X_iY_i]=1-1/q$. Orthogonality across distinct
subsets gives
\[
  \|F_\ell-G\|_2^2
  =\sum_{|I|\le\ell}c_I^2
    \left(1-\left(1-\frac1q\right)^{|I|}\right)
  \le\frac{\ell}{q}.
\]
Cauchy--Schwarz and triangle inequality give
\begin{equation}
    \E_{\H,f}|G_\H(f)^2-F_\H(f)^2|\le\|G-F\|_2\,\|G+F\|_2\le\bigl(\|G-F_\ell\|_2+\|F_\ell-F\|_2\bigr)\bigl(\|G\|_2+\|F\|_2\bigr)
\end{equation}
where $\|F\|_2=\|F_\ell\|_2=1$ and $\|G\|_2\le\|F_\ell\|_2$ by orthogonality; therefore
\begin{equation}
    \E_{\H,f}|G_\H(f)^2-F_\H(f)^2| \leq e^{-\Omega(d)}+O(\sqrt{d/q}).
\end{equation}
Combining this estimate with \cref{eq:dqi-output-probability}, and using
$q=n+1$, gives
\[
  \Pr[\OPIval(f)\ge\tau]
  \ge1-e^{-\Omega(d)}-O(\sqrt{d/n})=1-o(1).
\]
The sparse error, syndrome, and reusable decoding workspace each use
$O(d\log n)$ qubits. This proves the theorem.
\end{proof}

\end{document}